\documentclass{amsart}
\usepackage{graphicx,epstopdf}
\newcommand{\Cl}{\mathop{\hbox{\rm Cl}}\nolimits}
\newcommand{\bz}{\mathbb{Z}}
\newcommand{\br}{\mathbb{R}}

\newcommand{\zt}{\bz_2}
\newcommand{\rp}{\br\mathbb{P}}
\newcommand{\ch}{\mathcal{H}}
\newcommand{\wt}{\mathop{\text{\rm wt}}\nolimits}

\newcommand{\Hom}{\mathop{\text{\rm Hom}}\nolimits}
\newcommand{\Repl}{\mathop{\text{\rm Repl}}\nolimits}
\newtheorem{algorithm}{Algorithm}
\newtheorem{theorem}{Theorem}[section]
\newtheorem{definition}{Definition}[section]
\newtheorem{proposition}[theorem]{Proposition}
\newtheorem{corollary}[theorem]{Corollary}
\newtheorem{lemma}[theorem]{Lemma}
\newtheorem{example}[theorem]{Example}
\newcommand{\ack}{{\bf Acknowledgments}\\}

\begin{document}
 \title[An Application of Cubical Cohomology to Adinkras]{An application of Cubical Cohomology to Adinkras and Supersymmetry Representations}
 \author{C.F. Doran}
 \email{doran@math.ualberta.ca}
 \address{Department of Mathematical and Statistical Sciences,
 University of Alberta,
 632 CAB,
 Edmonton, AB T6G 2G1, Canada}

 \author{K.M. Iga}
 \email{kiga@pepperdine.edu}
 \address{Natural Science Division,
      Pepperdine University,
      Malibu, CA 90263}
      
 \author{G.D. Landweber}
 \email{gregland@bard.edu}
 \address{Mathematics Department,
     Bard College,
     Annandale-on-Hudson, NY 12504-5000}





\begin{abstract}
An Adinkra is a class of graphs with certain signs marking its vertices and edges, which encodes off-shell representations of the super Poincar\'e algebra.  The markings on the vertices and edges of an Adinkra are cochains for cubical cohomology.  This article explores the cubical cohomology of Adinkras, treating these markings analogously to characteristic classes on smooth manifolds.
\end{abstract}


\maketitle

\ack

This project arose out of questions that arose in discussions with M.~Faux, S.~J.~Gates, and T.~H\"ubsch, and consequent discussions with them helped in clarifying the presentation.  T.~H\"ubsch especially spent a lot of time giving careful and constructive comments on the paper that were very helpful.  Y.~Zhang, a graduate student of R.~Stanley at MIT, independently found the usefulness of describing the dashings as a 1-cochain, and discussions with him were fruitful.  He also made a number of helpful suggestions on notation.  He has written another paper that uses these ideas for a different purpose\cite{Yan}.

\section{Introduction}

Supersymmetry is a feature of some theories of interest in high energy particle physics.  It posits that the physical fields in nature form a representation of an algebra called the $N$-extended super Poincar\'e algebra, which is defined given two positive integers $N$ and $D$ (the number $D$ represents the dimension of space-time)\cite{WB}.

In recent years it has become apparent that even for $D=1$, the representation theory is surprisingly intricate\cite{rBG,rCRT,rA,rGR-1,rGR1,rKRT,rPT,rT01,rT01a,rT06}.  In 2004, M. Faux and S. J. Gates introduced a graphical diagram called an {\em Adinkra} to describe some commonly studied representations of the super Poincar\'e algebra in $D=1$ dimensions\cite{rA}.  These diagrams are directed graphs, where the vertices and edges have certain markings and colorings.  An Adinkra is similar in spirit to the concept of a Cayley graph or a Schreier graph in combinatorial group theory (see, for instance, Sec. 1.6 of \cite{mks}).  These Adinkra diagrams have broadened our knowledge of the representation theory of the super Poincar\'e algebra in one dimension\cite{r6-1,r6-2,r6-4,r6-forevery,rAT0,rKT07,toppan2011chiral}.

In this paper, we will show how a certain combinatorial cohomology theory on the Adinkra, called cubical cohomology, can be used to determine which kinds of markings are possible.

\subsection{Marked Graphs and Adinkras}
To define an Adinkra, it is first necessary to describe the kinds of markings to be put on vertices and edges.

A directed graph consists of $V$, a finite set of vertices, and a set of directed edges $E\subset V\times V$.  If $e$ is an edge, and $e=(v,w)$, then we say $e$ points from $v$ to $w$, and $e$ is incident to both $v$ and $w$.

A {\em bipartition} is a partition of $V$ into $V_0\sqcup V_1$, so that each edge of the graph is incident with one element of $V_0$ and one element of $V_1$.  The elements of $V_0$ are called {\em bosons} and the elements of $V_1$ are called {\em fermions}.  We draw the elements of $V_0$ as empty circles, and the elements of $V_1$ as filled circles.  A graph together with a bipartition is called a {\em bipartite graph}.

Fix a positive integer $N$.  We choose $N$ colors, one for each of the numbers in $\{1,\ldots,N\}$.  An {\em edge coloring} is a mapping $\chi:E\to \{1,\ldots,N\}$.  We draw each edge $e\in E$ using the color corresponding to $\chi(e)$.

A {\em dashing} is a mapping $\mu:E\to \{0,1\}$.  The edge $e\in E$ is drawn with a solid line if $\mu(e)=0$ and drawn with a dashed line if $\mu(e)=1$.\footnote{This is equivalent to the notion of a {\em signed graph}, defined by Frank Harary\cite{Harary}, where $+1$ and $-1$ are used instead of $0$ and $1$.  The difference is that we will use additive notation instead of multiplicative notation for the group of order 2.  The relationship of dashed edges with the supersymmetry representation actually suggests the multiplicative notation, and it may be that future work will motivate a switch to this notation, but since the purpose of this paper is to establish the relationship to cohomology, the additive notation will be used.}

In this paper, we use the term {\em marked graph} to mean a directed graph, together with a bipartition, an edge coloring, and a dashing.  An example of a marked graph is Example~\ref{ex:n2}.

Given a marked graph and a subset $J \subset\{1,\ldots,N\}$, we can take the subgraph given by the same vertex set, but for the edge set take $E'=\chi^{-1}(J)$; that is, the edges whose colors are in $J$.  The bipartition of the vertices will be the same and the colorings and dashing of the edges will be restricted to $E'$.  Connected components of this marked subgraph will be called $J$-color-faces of the marked graph.  If $k$ is an integer with $0\le k\le N$, a $k$-face is a $J$-color-face where $\#J=k$ (the cardinality of a set $J$ is denoted by $\#J$).

A {\em circuit} is a connected marked graph where each vertex is incident to exactly two edges.  A {\em quadrilateral} is a circuit with 4 vertices.

An {\em Adinkra} is a marked bipartite graph with the following properties\cite{rA}:

\begin{enumerate}
\item (Color-regular) For every vertex $v$ and color $j\in \{1,\ldots,N\}$, there is a unique edge incident with $v$ of color $j$.
\item (Square) Every 2-face is a quadrilateral.
\item (Non-escheric) Suppose $f$ is a 2-face.  In traversing the boundary of $f$ in either direction, the same number of edges are oriented along the direction of the path as are oriented against the path.\footnote{The term {\em non-escheric} is a reference to what was called an ``escheric Adinkra'' in \cite{rA} that satisfied a modified super Poincar\'e algebra with central charge, which itself is a reference to the M.~C.~Escher drawing, {\em Ascending and Descending}.}
\item (Odd) Every 2-face has an odd number of dashed edges.\footnote{In the language of Harary\cite{Harary}, every 2-face is negative.}
\end{enumerate}

\begin{example}\label{ex:n2}
Consider the following Adinkra for $N=3$:
\begin{center}
\begin{picture}(100,100)(-5,-5)
\put(0,0){\includegraphics[height=80pt]{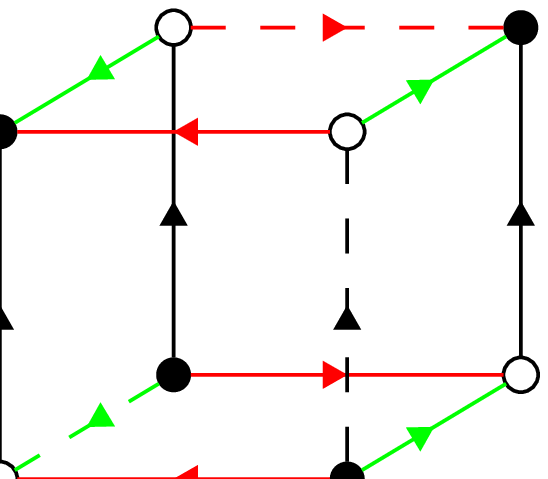}}
\put(-4,-6){\makebox[0in][r]{$v_1$}}
\put(-4,64){\makebox[0in][r]{$v_2$}}
\put(64,-6){\makebox[0in][l]{$v_3$}}
\put(66,58){\makebox[0in][l]{$v_4$}}
\put(18,22){\makebox[0in][l]{$v_5$}}
\put(19,82){\makebox[0in][l]{$v_6$}}
\put(86,12){\makebox[0in][l]{$v_7$}}
\put(86,82){\makebox[0in][l]{$v_8$}}
\end{picture}
\end{center}
\end{example}

Given an Adinkra, and a vertex $v$ in the Adinkra, a {\em vertex switch} produces another Adinkra identical to the first, except that every edge incident with $v$ changes dashedness; i.e., if it was dashed, it becomes solid, and vice versa.  A vertex switch does not change the oddness of the dashing since every 2-face either does not contain the vertex being switched (in which case nothing happens to the edges in it) or it does, in which case two edges are switched.  Therefore, the result is an Adinkra.

\begin{example}
If we do a vertex switch on $v_4$ in the previous example, we get the following Adinkra.  Note that each 2-face still has an odd number of dashed edges.
\begin{center}
\begin{picture}(100,100)(-5,-5)
\put(0,0){\includegraphics[height=80pt]{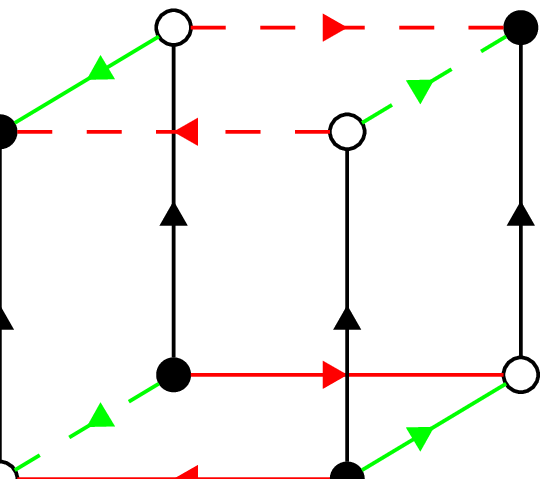}}
\put(-4,-6){\makebox[0in][r]{$v_1$}}
\put(-4,64){\makebox[0in][r]{$v_2$}}
\put(64,-6){\makebox[0in][l]{$v_3$}}
\put(66,58){\makebox[0in][l]{$v_4$}}
\put(18,22){\makebox[0in][l]{$v_5$}}
\put(19,82){\makebox[0in][l]{$v_6$}}
\put(86,12){\makebox[0in][l]{$v_7$}}
\put(86,82){\makebox[0in][l]{$v_8$}}
\end{picture}
\end{center}
\end{example}

\subsection{Purpose of this paper}
The purpose of this paper is to characterize the odd dashings on an Adinkra, up to vertex switches.  The machinery to understand this is cubical cohomology, which is a combinatorial counterpart to simplicial cohomology, but using cubes instead of simplices.

This paper fits into a program to classify all Adinkras, as follows.  Since the disjoint union of Adinkras is an Adinkra, it suffices to describe connected Adinkras.  If we ignore the direction of the edges and the dashings, so that we are considering only the underlying undirected graph with the edge colorings, then these are obtained by taking the quotient of a cube by a binary doubly-even code\cite{rAT0}.  The ways of assigning the directions of the edges were explored in the ``Hanging Gardens'' theorem\cite{r6-1}.  This paper solves the problem of identifying the ways in which the dashings may be chosen, up to vertex switches.

\subsection{Review of Codes}\label{sec:codes}
In this paper all codes will be binary linear block codes of length $N$.  That is, they are linear subspaces of $(\zt)^N$ as vector spaces over $\zt$.  We review a few basic concepts in the theory of codes to establish notation and conventions.  The interested reader can find a thorough introduction to the theory of codes in \cite{HP}.

If $C\subset (\zt)^N$ is a code, then elements of $C$ are called {\em codewords}.  If a codeword $w=(w_1,\ldots,w_N)\in(\zt)^N$, the {\em weight} of $w$, denoted $\wt(w)$, is the number of $j$ with $w_j=1$.  An {\em even} code is one where all codewords have even weights.  A {\em doubly even} code is one where all codewords have weights divisible by $4$.

A basis for $C$ is called a {\em generating set}.  A generating set is not unique, but all generating sets have the same cardinality: the {\em dimension} of the code, denoted as $k$, which is the dimension of $C$ as a vector space.    If a code is a binary linear block code of length $N$ of dimension $k$, it is called an $[N,k]$ code.

A generating set is often arranged as rows of a matrix, called the {\em generating matrix}.  Given two generating matrices for $C$, there is a sequence of row operations taking one generating matrix to another.

\section{Cubes}
Throughout this paper, we will be using cohomology with $\zt$-coefficients.  Thus we can ignore orientation issues.

The $N$-dimensional cube (also called an $N$-cube) is the space $I^N=[0,1]^N$.  It has a natural CW complex structure with $3^N$ different cells, one for every element of the set $\{0,*,1\}^N$.  The $*$ is a formal symbol that represents the interval $[0,1]$.  More specifically, if $(x_1,\ldots,x_N)\in\{0,*,1\}^N$, define
\[
Y_j=
\begin{cases}
[0,1],&\mbox{if $x_j=*$}\\
\{x_j\},&\mbox{if $x_j\not=*$.}
\end{cases}
\]
Then the geometric realization of $(x_1,\ldots,x_N)$ is $\prod_{j=1}^N Y_j$.  Let $J=\{j\,|\,x_j=*\}$.  The cardinality of $J$ is the dimension of the cell.  For instance, if $N=5$, then $(*,1,*,0,1)$ represents a 2-dimensional cell with $J=\{1,3\}$ and has geometric realization $[0,1]\times\{1\}\times[0,1]\times\{0\}\times\{1\}$.

Given a cell $(x_1,\ldots,x_N)\in\{0,*,1\}^N$, an integer $j\in \{0,\ldots,N\}$, and $\alpha\in\{0,*,1\}$, we define the replacement operator
\[
\Repl_{j,\alpha}(x_1,\ldots,x_N) = (x_1,\ldots,x_{j-1},\alpha,x_{j+1},\ldots,x_N).
\]

There is also an undirected graph whose vertices are the 0-cells (that is, the elements of $\{0,1\}^N$) and whose edges are the 1-cells (those elements of $\{0,*,1\}^N$ that have exactly one $*$).  If $(x_1,\ldots,x_N)$ is an edge, and $j$ is the unique coordinate so that $x_j=*$, then this edge connects the vertices $\Repl_{j,0}(x_1,\ldots,x_N)$ and $\Repl_{j,1}(x_1,\ldots,x_N)$.  We color such an edge with color $j$.  More generally, if $(x_1,\ldots,x_N)$ is any element of $\{0,*,1\}^N$, let $J=\{x_j|x_j=*\}$.  Then $(x_1,\ldots,x_N)$ is a $J$-color-face.  We can form a bipartition on the vertices by saying that $(x_1,\ldots,x_N)$ is a boson if and only if an even number of the $x_j$ are $1$.

Given a subset $J\subset\{1,\ldots,N\}$, define $\Psi_J$ to be the unique monotonic bijection from $\{1,\ldots,\#J\}$ to $J$.  Given a $J$-color-face $f$, and $n\in\{1,\ldots,\#J\}$, and $\alpha\in\{0,1\}$, define $\partial_n^\alpha f = \Repl_{\Psi_J(n),\alpha}(f)$.  If $m<n$, then
\[
\partial_m^\alpha \circ \partial_n^\beta = \partial_{n-1}^\beta \circ \partial_m^\alpha.
\]

In this way, the $N$-cube $I^N$ can be viewed as a cubical set, in analogy with simplicial sets.  Cubical sets appear in the work of Kan\cite{Kan}, though in those papers, the definition included degeneracy maps, and here those degeneracy maps are not needed.

For each integer $k$ define the group of $k$-chains $C_k(I^N)$ to be the free $\zt$-vector space spanned by the $k$-dimensional cells of $I^N$.  If $\sigma$ is a $k$-face, we define
\[\partial \sigma = \sum_{n=1}^k \sum_{\alpha=0}^1 \partial_n^\alpha(\sigma)\]
This operator is extended linearly to all of $C_k(I^N)$ for each $k$ (for $\sigma\in C_0(I^N)$ we define $\partial\sigma = 0$).  It is straightforward to show that $\partial\partial c= 0$ for all $c$.

It was proved in \cite{rAT0} that every connected Adinkra is a quotient of the graph of $I^N$ by a binary linear code $C$ of length $N$.  Specifically, the quotients are realized by compositions of reflection maps of the following type.

For every $1\le j\le N$, define the reflection map $\rho_j:\{0,*,1\}^N\to\{0,*,1\}^N$:
\[
\rho_j(x_1,\ldots,x_N)=\left\{\begin{array}{ll}
(x_1,\ldots,x_N),&x_j=*\\
\Repl_{k,1-x_j}(x_1,\ldots,x_N),&x_j\not=*.
\end{array}\right.
\]

Let $C$ be a code of length $N$.  For each $(t_1,\ldots,t_N)\in C$, we define
\[
\rho_{(t_1,\ldots,t_N)}=\rho_1^{t_1}\circ\cdots\circ \rho_N^{t_N}.
\]
In this way, $C$ acts on the cells of $I^N$, and in particular, its 1-skeleton.  The result in \cite{rAT0} is that every connected Adinkra has as its graph the quotient of the 1-skeleton of $I^N$ by a code $C$.  The map $\rho_x$ preserves the bipartition if and only if the weight of $x$ is even, so the quotient of an $N$-cube by a code will inherit the bipartition if and only if the code is even.  It was proven in \cite{rAT0} that a quotient of an $N$-cube by a code can be an Adinkra if and only if that code is doubly even.

The complex of cells in the quotient has the structure of a CW complex which we call a {\em cubical complex}.  Since the $\rho_x$ maps do not commute with the $\partial_k^\alpha$ (they may send $\partial^0$ to $\partial^1$, for instance), the result is not a cubical set, but as a CW complex, it is well-defined.

For every word $x\in C$ we extend $\rho_x$ linearly to a map $(\rho_x)_\#$ on chains.  Then it is straightforward to check that
\[
(\rho_x)_\#\circ \partial (\sigma) = \partial \circ (\rho_x)_\# (\sigma)
\]
by checking it explicitly for an arbitrary $k$-cell $\sigma$.  Thus, we can define chain groups $C_k(A)=C_k(I^N)/C$, and boundary homomorphisms $\partial:C_k(A)\to C_{k-1}(A)$.  Define the group of cycles $Z_k(A)$ to be the kernel of $\partial:C_k(A)\to C_{k-1}(A)$, and the group of boundaries $B_k(A)$ to be the image of $\partial:C_{k+1}(A)\to C_k(A)$.  Define the cubical homology of $A$ to be
\[
H_k(A)=Z_k(A)/B_k(A).
\]

Similarly, we define the cochain groups $C^k(A)=\Hom(C_k,\zt)$, or equivalently, the set of mappings from the set of $k$-cells of $A$ to $\zt$, with the structure of a $\zt$-vector space.  Define $d:C^k\to C^{k+1}$ so that $d\omega(c)=\omega(\partial c)$.  Again, $d^2=0$ and we define cocycle groups $Z^k(A)$ to be kernels of $d$, and coboundary groups $B^k(A)$ to be images of $d$.  Define the cubical cohomology of $A$ to be
\[
H^k(A)=Z^k(A)/B^k(A).
\]

Note that the definitions of $H_k(A)$ and $H^k(A)$ do not depend on the bipartition of the vertices or the coloring of the edges of $A$.

Given an Adinkra $A$, we define the following special cochains $\omega_k\in C^k(A)$: if $f$ is any $k$-face, let $\omega_k(f)=1$.  We now show that $d\omega_k=0$.  Let $\sigma$ be any $J$-cell with $|J|=k+1$.  We write (in $\zt$)
\begin{eqnarray*}
\nonumber
d\omega_k(\sigma) &=& \omega_k(\partial \sigma)\\
&=&\sum_{j=1}^k\sum_{\alpha=0}^1 \omega_k(\partial_i^\alpha\sigma)\\
&=&\sum_{j=1}^k\sum_{\alpha=0}^1 1\\
&=&2k\\
&=&0\pmod{2}.
\end{eqnarray*}
Therefore $\omega_k\in Z^k(A)$, and we let $w_k=[\omega_k]\in H^k(A)$ be its cohomology class.

\section{Dashed edges}
\label{sec:dash}
A dashing on a graph can be recorded by a 1-cochain $\mu\in C^1(A)$ that assigns 1 to every dashed edge and 0 to every solid edge.  

\begin{theorem}
Let $A$ be a quotient of a cube by a code.  There exists an odd dashing on $A$ if and only if $w_2=0$.\label{thm:w2}
\end{theorem}

\begin{proof}
For a dashing to be odd means for every 2-cell $f$, there must be an odd number of dashed edges around $f$.  If $\mu$ represents this dashing, then this is equivalent to the statement that $\mu(\partial f)=1$ for all 2-cells $f$.  This, in turn, is equivalent to $d\mu(f)=1$ for all 2-cells $f$; in other words, that $d\mu=\omega_2$.  Since $d\omega_2=0$, the existence of such a $\mu$ is equivalent to whether or not $w_2=0$ in $H^2(A)$.
\end{proof}

\begin{example}
When $A$ is an $N$-dimensional cube, then since the geometric realization of $A$ is a CW complex, cubical cohomology coincides with ordinary singular cohomology; and since $I^N$ is contractible, $H^2(A)=0$.  Therefore $w_2=0$, and an odd dashing exists.
\end{example}

Remark:  Since previous work\cite{rAT0} has shown that the existence of an odd dashing is equivalent to the code $C$ being doubly even, it must also be true that the condition $w_2=0$ is equivalent to the code $C$ being doubly even.  As it happens, this can be proved directly in the following way: first identify the generators for $H_2(A)$.  These will be certain unions of 2-faces in the Adinkra and will be determined by the codewords of $C$.  Applying $w_2$ to these generators simply counts the 2-faces.  This will be zero modulo 2 if and only if the codeword in question is weight zero modulo 4.

The next question is that of uniqueness.  There might indeed be two 1-cochains $\mu_1$ and $\mu_2$, with $d\mu_1=\omega_2=d\mu_2$.  These are both odd dashings.  For instance, a vertex switch will change one odd dashing to another odd dashing.

\begin{proposition}
Let $A$ be a cubical complex, and $\mu$ an odd dashing.  Let $T$ be any set of vertices of $A$.  Let $f_T$ be the $0$-chain so that $f_T(v)=1$ for all $v\in T$ and $f_T(v)=0$ otherwise.  Then the dashing that results from $\mu$ of doing vertex switches for all vertices in $T$ is $\mu+df_T$.
\label{prop:vertexsignflip}
\end{proposition}

\begin{proof}
When $T$ consists of a single vertex $v$, $df_T$ is 1 for all edges adjacent to $v$.  Thus, $\mu+df_T$ differs from $\mu$ precisely for those edges that are adjacent to $v$.

The general proof then follows from iterating this procedure over all the elements of $T$.
\end{proof}

\begin{theorem}
Let $A$ be a cubical complex.  Two odd dashings $\mu_1$ and $\mu_2$ are related by a series of vertex switches if and only if $\mu_2-\mu_1=df$ for some $0$-cochain $f$.  Therefore, the set of odd dashings, modulo vertex switches, is in one-to-one correspondence with $H^1(A)$.\label{thm:classifydashings}
\end{theorem}
 
\begin{proof}
Proposition~\ref{prop:vertexsignflip} shows that if a series of vertex switches is done on $\mu_1$, then $\mu_2=\mu_1+df$.

Conversely, if $\mu_2=\mu_1+df$, then by Proposition~\ref{prop:vertexsignflip} the vertex switch based on the set $f^{-1}(\{1\})$ on $\mu_1$ produces $\mu_2$.

Fix an odd dashing given by the 1-cochain $\mu_0$, so that $d\mu_0=\omega_2$.  Any odd dashing $\mu$ has $d\mu=\omega_2$, so $\mu-\mu_0$ is a 1-cocycle.  The mapping  $\mu\mapsto [\mu-\mu_0]\in H^1(A)$ is thus an affine-linear bijection from the set of odd dashings modulo vertex switches to $H^1(A)$.
\end{proof}

Remark: This is remarkably similar to the situation with spin structures on manifolds: the existence of a spin structure on an orientable manifold $M$ is determined by the second Stiefel--Whitney class of the tangent bundle $w_2(TX)\in H^2(M;\zt)$.  A necessary and sufficient condition for spin structures to exist is $w_2(TX)=0$, and the difference of two such spin structures is characterized by an element in $H^1(M;\zt)$, so that the set of spin structures is in bijection with $H^1(M;\zt)$.  It is not clear whether the similarity with the criterion for dashed edges is significant, or merely coincidental.  One possible connection is due to Cimasoni and Reshetikhin, where edge dashings on surface graphs correspond to spin structures\cite{CimResh}.

\begin{example}
Since we know that $H^1(I^N)=0$, it follows that for $N$-dimensional cubes, there is only one odd dashing up to vertex switches.
\end{example}

\begin{example}
\label{ex:d4}
We now consider a quotient of a $4$-cube by the code
\[d_4=\{(0,0,0,0),(1,1,1,1)\}.\]
The graph is obtained by taking two copies of the $3$-cube (colored black, red, and green), connected by blue lines representing the fourth color.  When we identify antipodal vertices, we are left with only one $3$-cube, and the blue lines connect antipodal vertices in the $3$-cube.  The result is the following:

\begin{center}
\begin{picture}(90,90)(0,0)
\put(0,0){\includegraphics[height=80pt]{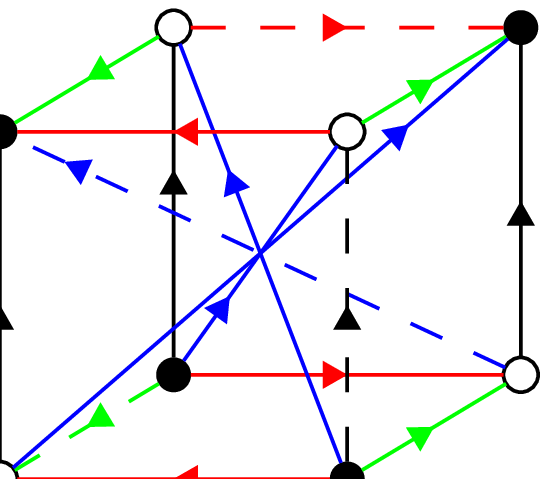}}
\end{picture}
\end{center}

As a CW complex, the 4-cube is a 4-ball, with its boundary 3-sphere as its 3-skeleton.  The quotient acts on the 3-sphere by antipodal mapping, and sends the 4-cell to itself in a way that the boundary is two copies of the 3-sphere.  Thus, the corresponding chain complex is that of the real projective space $\rp^4$.  It is known that $H^1(\rp^4;\zt)\cong \zt$, so there are two odd dashings up to vertex sign flips.  These correspond to the chiral\footnote{More precisely, the dimensional reduction to $D=1$ of the $D=4$ chiral multiplet} and twisted chiral multiplets.  We will show this in Example~\ref{ex:d4s}, once we have the more general theory in place. 
\end{example}

\section{The Adinkra cohomology of a quotient of a cube}
\subsection{Computation of Adinkra cohomology of a quotient of a cube}
\label{sec:adinkracohomology}
In this section we compute some homology and cohomology groups of a quotient of an $N$-cube by a binary linear code $C$.

\begin{theorem}
\label{thm:homology}
Let $A$ be a quotient of an $N$-cube by a binary linear $[N,k]$ block code $C$.  Suppose the minimal weight of a non-zero codeword in $C$ is greater than or equal to 4.  Then we have the following.
\begin{eqnarray*}
H_0(A)&\cong&\zt\\
H_1(A)&\cong&(\zt)^k\\
\\
H^0(A)&\cong&\zt\\
H^1(A)&\cong&(\zt)^k
\end{eqnarray*}
\end{theorem}

\begin{proof}
When $(w_1,\ldots,w_N)\in C$ is a codeword, the reflection $\rho_w$ on the $N$-cube has fixed points inside cells $(x_1,\ldots,x_N)$ where $x_i=*$ whenever $w_i=1$.  Since the minimal weight of a non-zero codeword in $C$ is greater than or equal to 4, the fixed points of the action of $C$ on $I^N$ lie only in cells of dimension 4 and higher.  Thus, $C$ acts freely on the 3-skeleton of $I^N$ (denoted $I^{N(3)}$), and therefore, $I^{N(3)}/C$ is a geometric realization for $A$ up to dimension $3$.

The computations of $H_k(A)$ and $H^k(A)$ for $k\le 2$ involve only the cells up to dimension $3$, and since $I^{N(3)}/C$ is the CW complex corresponding to the cells up to dimension $3$ in $A$, we have for $k\le 2$:
\[
H_k(A)\cong H_k(I^{N(3)}/C;\zt)
\]
and
\[
H^k(A)\cong H^k(I^{N(3)}/C;\zt).
\]
Since $I^{N(3)}$ is connected, this implies the result for $H_0(A)$ and $H^0(A)$.

Since $I^{N(3)}$ is a covering space of $I^{N(3)}/C$ with covering group $C$, we have that $\pi_1(I^{N(3)}/C)\cong C$.  By the Hurewicz theorem and the group isomorphism $C\cong (\zt)^k$, this provides $H_1(I^{N(3)}/C)\cong (\zt)^k$.  The universal coefficient theorem can be used to show that
\[
H^1(I^{N(3)}/C)\cong (\zt)^k.
\]
\end{proof}

Remark: It is possible to prove that $H^2(A)\cong (\zt)^{k+\binom{k}{2}}$ through a more involved argument, but that computation will not be used in this paper.

\subsection{Cohomology generators $\epsilon_j$}
In the previous section, we calculated the first cohomology group of an Adinkra $A$.  But it will be useful to have specific generators.  In this section, we define such generators, and these definitions do not depend on $C$ being doubly even.  The fact that these generate $H_1$ does depend on the previous section, however.  Also note that these are not, generally speaking, linearly independent.

We define $\epsilon_1,\ldots,\epsilon_N$ as follows: let $\epsilon_j$ be the 1-cochain that is 1 on all edges of color $j$, and 0 on all other edges.

\begin{proposition}
The 1-cochain $\epsilon_j$ is a 1-cocycle.
\end{proposition}

\begin{proof}
Let $\sigma$ be any 2-face.  Then
\begin{eqnarray*}
d\epsilon_j(\sigma)&=&\epsilon_j(\partial\sigma)\\
&=&\sum_{n=1}^2 (\epsilon_j(\partial_n^1\sigma)+\epsilon_j(\partial_n^0\sigma))
\end{eqnarray*}
The edge $\partial_n^1\sigma$ is of color $j$ if and only if $\partial_n^0\sigma$ is.  This implies that this sum is zero modulo 2.
\end{proof}

We will sometimes use the notation $\epsilon_j$ to denote the cohomology class represented by $\epsilon_j$.

\subsection{Lifting paths and squares}
To represent homology classes in $A$ it will be convenient to have a combinatorial version of path lifting.
\begin{lemma}
Let $A$ be a graph with edge coloring that is color-regular.
Suppose $v_0$ is any vertex of $A$.  Given a sequence of colors $j_1,\ldots,j_r$, there is a unique path starting at $v_0$ of length $r$ so that the $i$th edge in the path has color $j_i$.\label{lemma:pathlifting}
\end{lemma}

\begin{proof}
This can be proved by induction on the length $r$ of the sequence of colors, using the color-regularity property.
\end{proof}

Likewise, it will be useful to lift squares, and so we prove the following lemma:
\begin{lemma}
Let $A$ be a graph with edge coloring that is color-regular, and for which the 2-faces are quadrilaterals.  Suppose $v_0$ is any vertex of $A$.  Given two different colors $i$ and $j$, and the graph $G$ with edge colors
\begin{center}
\begin{picture}(55,60)(-5,-10)
\put(0,0){\line(1,0){40}}
\put(0,0){\line(0,1){40}}
\put(0,40){\line(1,0){40}}
\put(40,0){\line(0,1){40}}
\put(0,0){\circle*{4}}
\put(40,0){\circle*{4}}
\put(0,40){\circle*{4}}
\put(40,40){\circle*{4}}
\put(20,-10){\makebox[0in][c]{$i$}}
\put(20,42){\makebox[0in][c]{$i$}}
\put(-3,18){\makebox[0in][r]{$j$}}
\put(42,18){\makebox[0in][l]{$j$}}
\put(-10,-10){$P$}
\put(42,-10){$Q$}
\put(42,42){$R$}
\put(-10,42){$S$}
\end{picture}
\end{center}
there is a unique mapping of $G$ to $A$ where $P$ is sent to $v_0$ and edges are sent to edges of the same color.  There is a unique $\{i,j\}$-color-face in $A$ with this boundary.
\label{lemma:squarelifting}
\end{lemma}

\begin{proof}
We first view this square as a path with sequence of colors $i$, $j$, $i$, $j$:

\begin{center}
\begin{picture}(160,24)(0,-12)
\put(0,0){\line(1,0){40}}
\put(40,0){\line(1,0){40}}
\put(80,0){\line(1,0){40}}
\put(120,0){\line(1,0){40}}
\put(0,0){\circle*{4}}
\put(40,0){\circle*{4}}
\put(80,0){\circle*{4}}
\put(120,0){\circle*{4}}
\put(160,0){\circle*{4}}
\put(0,-12){\makebox[0in][c]{$P$}}
\put(40,-12){\makebox[0in][c]{$Q$}}
\put(80,-12){\makebox[0in][c]{$R$}}
\put(120,-12){\makebox[0in][c]{$S$}}
\put(160,-12){\makebox[0in][c]{$T$}}
\put(20,5){\makebox[0in][c]{$i$}}
\put(60,5){\makebox[0in][c]{$j$}}
\put(100,5){\makebox[0in][c]{$i$}}
\put(140,5){\makebox[0in][c]{$j$}}
\end{picture}
\end{center}

By Lemma~\ref{lemma:pathlifting}, there is a unique path in $A$ that lifts this path where $P$ goes to $v_0$.  Since there are two colors, this path is part of an $\{i,j\}$-color-face $\sigma$.  Since the 2-faces of $A$ are quadrilaterals, the path must close, so that $P$ and $T$ go to the same vertex in $A$.  Thus, the square $G$ maps into $A$.  The boundary of $\sigma$ is that path.
\end{proof}

\subsection{$H_1$-generators $P_c$}\label{sec:pgenerators}
Theorem~\ref{thm:homology} says that $H_1(A)$ is isomorphic to the code $C$, but it will be useful to have a more constructive description of this isomorphism.

\begin{definition}
Pick a vertex $v_0$ from the Adinkra $A$.  Let $w=(w_1,\ldots,w_N)\in C$ be a word of length $N$ and weight $r$.  Let $1\le j_1<j_2<\cdots<j_r\le N$ be the integers so that $w_{j_i}=1$.  

We use Lemma~\ref{lemma:pathlifting} to find a path ${}^{v_0}P_w$ in $A$ starting from $v_0$ so that the $i$th edge is of color $j_i$ for all $i$.  We interpret ${}^{v_0}P_w$ as a sum in $C_1(A)$ of its edges.  When $v_0$ is understood, we will write $P_w$.
\end{definition}

\begin{example}
Example~\ref{ex:d4} has a code $\{(0,0,0,0),(1,1,1,1)\}$.  Take the codeword $w=(1,1,1,1)$, so that $j_1<j_2<j_3<j_4$ is $1<2<3<4$.  Pick the bosonic vertex on the lower left of the diagram.  The ordering of the $j_i$ means we go up from this vertex along color 1 (black), then to the right along color 2 (red), then to the back face using color 3 (green), then across the diagonal along color 4 (blue).  Note that this brings us back to the starting point.  This path is $P_{(1,1,1,1)}$.  The interested reader may verify that a similar path, starting at any other vertex, will also form a loop.
\label{ex:path}
\end{example}

\begin{proposition}
Fix a vertex $v_0$ in $A$.  For every codeword $w\in C$, ${}^{v_0}P_w$ is a 1-cycle.  The mapping $\Pi\colon C\to H_1(A)$ sending $w\in C$ to the homology class $[{}^{v_0}P_w]\in H_1(A)$ is a monomorphism, and if the minimal weight of a non-zero codeword in $C$ is at least 3, then $\Pi$ is an isomorphism.
\label{prop:pi}
\end{proposition}

\begin{proof}
There is a vertex $x=(x_1,\ldots,x_N)$ in $\{0,1\}^N$ that goes to $v_0$ under the quotient of $I^N$ by $C$.  Use Proposition~\ref{lemma:pathlifting} to find a path ${}^x\tilde{P}_w$ in $I^N$ starting at $x$ so that the $i$th edge is of color $j_i$.  Under the quotient by $C$, this path results in a similar path starting from $v_0$ so that the color of the $i$th edge is $j_i$ for all $i$.  By the uniqueness of such a path, this must be ${}^{v_0}P_w$.  Since ${}^x\tilde{P}_w$ ends at $x+w$ and $w\in C$, it must be that ${}^{v_0}P_w$ ends at $v_0$.  Therefore, ${}^{v_0}P_w$ is a cycle.

We now show that the map $\Pi$ is a homomorphism.  Suppose $v$ and $w$ are two words in $C$, with corresponding sequence of colors $i_1<\ldots<i_q$ and $j_1<\ldots<j_r$, respectively.

Consider the color sequence $i_1,\ldots,i_q,j_1,\ldots,j_r$.  Write this sequence as $k^0_1,\ldots,k^0_s$.  We obtain a path $F^0$ starting at $v_0$ following this sequence of colors.  This path consists of a sequence of vertices $v^0_0,\ldots,v^0_s$ and edges $e^0_1,\ldots,e^0_s$ so that for each $i$, $e^0_i$ is incident with $v^0_{i-1}$ and $v^0_i$, and is of color $k^0_i$.  As a homology class, $[F^0]=[P_v]+[P_w]$.

At each stage $t$, we perform a sequence of adjacent swaps, that is, for some index $1\le n<s$, we define
\[
k^t_i=
\begin{cases}
k^{t-1}_n,&\mbox{if $i=n+1$}\\
k^{t-1}_{n+1},&\mbox{if $i=n$}\\
k^{t-1}_i,&\mbox{otherwise.}
\end{cases}
\]
The choice of $n$ at each stage $t$ is given by the bubble sort algorithm to obtain for some $t$ a sequence $k^t_i$ that is non-decreasing.  We can ensure that for no $t$ do we swap identical colors.

For each $t$, we produce a path $F^t$ starting at $v_0$ and following the colors in $k^t_i$.  More specifically, $F^t$ consists of a sequence of vertices $v^t_0=v_0,\ldots,v^t_s$ and edges $e^t_1,\ldots,e^t_s$ with $e^t_i$ incident to $v^t_{i-1}$ and $v^t_i$ and color $k^t_i$, for all $i$.  Consider the swap done at stage $t$, and let $n$ be the index such that $k^t_n=k^{t-1}_{n+1}$.  If we consider the vertex $v_{n-1}$ in $F^{t-1}$, and let $S$ be the square in $A$ starting at $v_{n-1}$ with colors $k^t_n$ and $k^t_{n+1}$ (see Prop.~\ref{lemma:squarelifting}), then $S$ is a face so that $\partial S$ is the union of the edges $e^{t-1}_n, e^{t-1}_{n+1}, e^t_n, e^t_{n+1}$.  Thus, $F^t=F^{t-1}+S$.  In this way, by induction, $[F^t]=[F^0]$ for all $t$.

Let $t$ be the stage at which the $k^t_i$ are non-decreasing.  Then $\Pi^t$ consists of a path, possibly with backtracks when two adjacent colors are identical.  Note that runs of identical colors may be at most 2 edges long.  When there is a backtrack corresponding to a run of 2 identical colors, then as an element of $C_1(A)$, these two edges cancel, and the corresponding pair of colors may be deleted.  The result of removing such pairs shows that $\Pi^t$ is the path starting at $v_0$ with colors corresponding to the non-zero entries in the codeword $v+w$.  Since $C$ is a finite group, this suffices to prove that $\Pi$ is a homomorphism.

To establish that $\Pi$ has trivial kernel, let $w=(w_1,\ldots,w_N)$ be any non-zero codeword, and let $j$ be any index so that $w_j=1$.  Then
\[
\epsilon_j([P_w]) = \sum_{i=1}^r \epsilon_j(e_i)=1
\]
and therefore, $[P_w]\not=0$, and thus, $\Pi$ is a monomorphism.

If the minimal weight of a non-zero codeword in $C$ is at least 3, then Theorem~\ref{thm:homology} applies and $C$ and $H_1(A)$ are vector spaces of the same finite dimension over $\zt$.  Thus, $\Pi$ being a monomorphism implies it is an isomorphism.
\end{proof}

\section{Classification of odd dashings modulo vertex switches}
Suppose we have an Adinkra $A$ that is a quotient of an $N$-cube by a doubly even code $C$.  We know from Theorem~\ref{thm:classifydashings} that the set of odd dashings modulo vertex switches is an affine space modeled on $H^1(A)$, which, by Theorem~\ref{thm:homology}, is isomorphic to $(\zt)^N$.  We will now describe a particular isomorphism from the set of odd dashings modulo vertex switches to $(\zt)^N$.  The point is to describe a complete invariant $(s_1,\ldots,s_k)\in (\zt)^N$ for the set of odd dashings modulo vertex switches which can be used in computations.

\begin{definition}
Suppose $c_1,\ldots,c_k$ is a generating set for the code $C$.  Pick a vertex $v_0$ of $A$.  As in Section~\ref{sec:pgenerators}, for each $i$, let ${}^{v_0}P_{c_i}$ be the path starting from $v_0$ and corresponding to the colors indicated by $c_i$.  For each $\mu\in C^1(A)$ associated with an odd dashing, define ${}^{v_0}s(\mu)=(s_1,\ldots,s_k)\in(\zt)^k$ with $s_i$ equal to the number, modulo 2, of dashed edges of $\mu$ along ${}^{v_0}P_{c_i}$.  As before, when $v_0$ is understood, we can write $s(\mu)$.\footnote{More invariantly, but less convenient computationally, ${}^{v_0}s(\mu)$ could be viewed as a linear functional on the code $C$.}
\end{definition}

\begin{example}
In Example~\ref{ex:path}, we followed a path that involved zero dashed edges.  So $s$ in this case is 0.  The interested reader can try this with other starting points, and should note in this case that $s$ is 1 if the starting vertex is a fermion, and 0 if the starting vertex is a boson.
\end{example}

\begin{theorem}
The map $s$ from the set of odd dashings modulo vertex switches to $(\zt)^k$ is a bijection.
\label{thm:dashingfunction}
\end{theorem}

\begin{proof}
First, note that $s$ is unchanged under vertex switches, because
\[
s(\mu+df)=s(\mu)+df(P_{c_i})=s(\mu)+f(\partial P_{c_i})=s(\mu)+0.
\]
Therefore the domain of $s$ can be taken to be odd dashings modulo vertex sign switches.

Now note that $s$ is affine-linear.  That is, if we let $\mu_0$ be an odd dashing, then $s(\mu)-s(\mu_0)$ is a linear function of $\mu-\mu_0$.  The map $s$ is then a bijection if and only if the corresponding linear function is one as well.  Since the dimensions of the domain and range are equal, the only thing we need to check is that the linear map is an injection.

Suppose $s(\mu)-s(\mu_0)=0$; in other words, $\mu(P_{c_i})-\mu_0(P_{c_i})=0$ for all $i$.  By linearity, $\mu(P_c)-\mu_0(P_c)=0$ for all $c\in C$.  Therefore, $\mu-\mu_0=0$ in cohomology, and $\mu-\mu_0=df$ for some $f$ given by a sequence of vertex switches.
\end{proof}

Remark.  Even though the paths $[P_{c_i}]$ generate $H_1(A)$, the application of $\mu$ to such a 1-cycle is not well-defined on the homology class; that is, it may differ when a path is modified to something homologous.  The reason is that $\mu$ is not a 1-cocycle, since $d\mu$ is $\omega_2$, not 0.

Along these lines, note that the actual sequence $s_1,\ldots,s_k$ depends on $v_0$.  It also depends on the ordering of the colors when lifting the sequence of colors to $P_{c_i}$, which we chose by specifying $j_1<\cdots<j_r$.  But if these choices are fixed for a given Adinkra, then this theorem says that $s(\mu)$ is a complete invariant for dashings modulo vertex switches.  We will consider the effect of changing $v_0$ in Section~\ref{sec:changev0}.

\subsection{Edge sign flips}
To enumerate all Adinkras, it will be useful to find all odd dashings modulo vertex switches for a given Adinkra.  This involves inverting the map $s$, meaning given a $k$-tuple $(s_1,\ldots,s_k)$, we wish to find an odd dashing $\mu$.  Actually, we assume that some other odd dashing $\mu_0$ is already given, with some $s(\mu_0)=(\sigma_1,\ldots,\sigma_k)$.  Then we describe what to do to $\mu_0$ to change $s(\mu_0)$ to any other $k$-tuple.  This will be achieved with what are called edge sign flips, described next.

If we are given an odd dashing, and $j$ is a color, then an {\em edge sign flip of color $j$} takes the dashing and reverses the dashing of all edges of color $j$.  An example of this was reversing the dashing on the blue edges in Example~\ref{ex:d4}. 

\begin{proposition}
Let $A$ be an Adinkra with odd dashing described by $\mu_1\in C^1(A)$.  Let $j$ be a color.  The result of an edge sign flip of color $j$ is a dashing described by $\mu_2=\mu_1+\epsilon_j$.  This resulting dashing is odd.  If $G$ is the generating matrix for the code, then $s(\mu_2)$ is $s(\mu_1)$ plus the $j$th column of $G$, taken modulo 2.
\label{prop:edgesignflip}
\end{proposition}

\begin{proof}
Recall that $\epsilon_j$ takes the value $1$ on every edge of color $j$, and $0$ on all other edges.  Therefore, adding it modulo 2 to $\mu_1$ reverses every edge of color $j$.  We recall that $d\epsilon_j=0$, so $d(\mu_1+\epsilon_j)=d\mu_1$, so $\mu_2$ is odd if $\mu_1$ was.  For each $i$ we see that $s_i(\mu_2)-s_i(\mu_1)=\epsilon_j(P_{c_i}) = (c_i)_j = G_{ij}$.
\end{proof}

More generally, we can consider a series of edge sign flips: for every word $x=(x_1,\ldots,x_N)\in(\zt)^N$ we perform an edge sign flip of every color $j$ where $x_j=1$.

\begin{theorem}\label{thm:edgesignflip}
Let $A$ be an Adinkra that is an $N$-cube quotiented by a code $C$, and suppose $A$ has odd dashing $\mu_0$.  Let $x\in (\zt)^N$ be a word of length $N$.  The result of performing a series of edge sign flips using the colors $j$ where $x_j=1$ on $\mu_0$ results in another odd dashing $\mu_x$.

If we write $x$ and $s(\mu)$ as column vectors, and $G$ is the generating matrix for the code, then
\begin{equation}
s(\mu_x)\equiv s(\mu_0)+Gx\pmod{2}.\label{eqn:matrixedgesignflip}
\end{equation}
\end{theorem}

\begin{proof}
The result of the series of edge sign flips using $x$ on $\mu_0$ is
\[
\mu_x\equiv \mu_0+\sum_{j=1}^N x_j \epsilon_j.
\]
If $w=(w_1,\ldots,w_N)\in C$ is a codeword, we take the above equation and evaluate it on the 1-chain $P_w$, and the result is
\[
s_a(\mu_x)\equiv s_a(\mu_0)+\sum_{j=1}^N x_j w_j \pmod{2}
\]
for all $a$.  In particular for the $i$th generating word of $G$,
\[w_j = G_{i,j},\]
and then (\ref{eqn:matrixedgesignflip}) follows from writing this statement in matrix form.
\end{proof}

We are now able to find an odd dashing for each $k$-tuple $(s_1,\ldots,s_k)$, using the following:
\begin{algorithm}
Input:
\begin{itemize}
\item a generating matrix $G$ for a code $C$
\item a vertex $v_0$ of the Adinkra $A=I^N/C$
\item an odd dashing $\mu_0$ on $A$
\item a $k$-tuple $(s_1,\ldots,s_k)\in (\zt)^N$
\end{itemize}

Output: an odd dashing $\mu$ on $A$ with ${}^{v_0}s(\mu)=(s_1,\ldots,s_k)$

Procedure:

Use Gauss--Jordan elimination on $G$.  Then $G$ is in reduced row-echelon form, with leading 1s in columns $j_1<\cdots<j_k$.  Compute ${}^{v_0}s(\mu_0)$.  Let
\[
(\sigma_1,\ldots,\sigma_k)=(s_1,\ldots,s_k)-s(\mu_0)\pmod{2}.
\]
Perform an edge sign flip on $\mu_0$ for color $j_i$ for every $i$ with $\sigma_i=1$.  Let $\mu$ be the resulting dashing after all these edge flips are accomplished.

By Theorem~\ref{thm:edgesignflip},
\[{}^{v_0}s(\mu)={}^{v_0}s(\mu_0)+(\sigma_1,\ldots,\sigma_k),\]
which is $(s_1,\ldots,s_k)$.
\end{algorithm}

Edge sign flips are powerful enough to generate all of the different odd dashings up to vertex switches, as we see in the following theorem.

\begin{corollary}
Let $A$ be a connected Adinkra, and let $\mu_1$ and $\mu_2$ be two odd dashings on $A$.  There is a sequence of edge sign flips and vertex switches that turns $\mu_1$ into $\mu_2$.
\end{corollary}

\begin{proof}
Every connected Adinkra is a quotient of an $N$-cube by a code $C$.  Let $G$ be a generating matrix for $C$.  Pick a vertex $v_0$ in the Adinkra.  Apply the above Algorithm to $G$, $v$, $\mu_1$, and ${}^{v_0}s(\mu_2)-{}^{v_0}s(\mu_1)$, to find an odd dashing $\mu_3$ with $s(\mu_3)=s(\mu_2)$.  By Theorem~\ref{thm:dashingfunction}, $\mu_2$ can be obtained from $\mu_3$ by a series of vertex switches.
\end{proof}

\section{Changing the starting vertex $v_0$}
\label{sec:changev0}
We now have a complete invariant, ${}^{v_0}s(\mu)$, of the odd dashing.  But it depends on the starting vertex $v_0$.  Let us consider the effect of changing $v_0$ on the $k$-tuple ${}^{v_0}s(\mu)$.  As before, we assume the Adinkra $A$ is connected, and therefore, it is a quotient of $I^N$ by a code $C$.

Let $v_1$ be another vertex in $A$.  There is a path $P$ connecting $v_0$ to $v_1$.  If we follow the sequence of colors in $P$, we form a word $x$, where $x_j=1$ if and only if color $j$ appears in the path $P$.

\begin{theorem}
Suppose $A$ is an Adinkra, with odd dashing given by $\mu$.  Let $v_0$, $v_1$, and $x$ be as above.  If we write $x$ and $s(\mu)$ as column vectors, and $G$ is the generating matrix for the code, then
\[
{}^{v_1}s(\mu)\equiv{}^{v_0}s(\mu)+Gx\pmod{2}.
\]
\label{thm:switchv0}
\end{theorem}

\begin{proof}
As in the proof of Theorem~\ref{thm:edgesignflip}, we actually prove that if $w=(w_1,\ldots,w_n)$ is a codeword in $C$, then
\begin{equation}
\mu\left({}^{v_1}P_w\right)
\equiv \mu\left({}^{v_0}P_w\right) + \sum_{j=1}^N x_j w_j\pmod{2}.
\label{eqn:switchv01}
\end{equation}
We take for $w$ each generator of $C$, that is, each row of the matrix $G$.  Gathering this into matrix notation, we then would have the desired result.

As before, let $j_1< \cdots < j_r$ be the integers so that $w_{j_i}=1$.

To prove (\ref{eqn:switchv01}), first consider the situation where the $x$ is all zeros except for a single 1 at position $m$.  That is, $x_m=1$ and $x_i=0$ for all $i\not=m$.

Case 1: $w_m=0$

If $w_m=0$, then consider the following diagram.
\begin{center}
\begin{picture}(280,60)(-10,-10)
\put(0,0){\line(1,0){100}}
\put(0,40){\line(1,0){100}}
\put(140,0){\line(1,0){100}}
\put(140,40){\line(1,0){100}}
\put(0,0){\line(0,1){40}}
\put(40,0){\line(0,1){40}}
\put(80,0){\line(0,1){40}}
\put(160,0){\line(0,1){40}}
\put(200,0){\line(0,1){40}}
\put(240,0){\line(0,1){40}}
\put(120,20){\makebox[0in][c]{$\cdots$}}
\put(0,0){\circle*{4}}
\put(0,40){\circle*{4}}
\put(240,0){\circle*{4}}
\put(240,40){\circle*{4}}
\put(-2,-5){\makebox[0in][r]{$v_0$}}
\put(242,-5){\makebox[0in][l]{$v_0$}}
\put(-2,42){\makebox[0in][r]{$v_1$}}
\put(242,42){\makebox[0in][l]{$v_1$}}
\put(20,-9){\makebox[0in][c]{$j_1$}}
\put(60,-9){\makebox[0in][c]{$j_2$}}
\put(180,-9){\makebox[0in][c]{$j_{r-1}$}}
\put(220,-9){\makebox[0in][c]{$j_r$}}
\put(20,45){\makebox[0in][c]{$j_1$}}
\put(60,45){\makebox[0in][c]{$j_2$}}
\put(180,45){\makebox[0in][c]{$j_{r-1}$}}
\put(220,45){\makebox[0in][c]{$j_r$}}
\put(-2,18){\makebox[0in][r]{$m$}}
\put(38,18){\makebox[0in][r]{$m$}}
\put(78,18){\makebox[0in][r]{$m$}}
\put(158,18){\makebox[0in][r]{$m$}}
\put(198,18){\makebox[0in][r]{$m$}}
\put(238,18){\makebox[0in][r]{$m$}}
\end{picture}
\end{center}
Use Lemma~\ref{lemma:squarelifting} iteratively to map this figure into $A$ with the bottom left vertex going to $v_0$, and the edges going to edges colored with the labels shown.  The leftmost edge connects $v_0$ to $v_1$.  Then the bottom edge of this diagram goes to the path ${}^{v_0}P_w$, and the top edge goes to the path ${}^{v_1}P_w$.  Both of these paths are circuits, and so the left edge of this diagram matches the right edge.  Let $F$ be the sum of the faces that this diagram gets sent to in $A$.  The boundary of $F$ is ${}^{v_0}P_w+{}^{v_1}P_w$.  Therefore we have (mod 2):
\begin{eqnarray*}
\mu({}^{v_1}P_w)
&\equiv& \mu({}^{v_0}P_w+\partial F)\\
&\equiv& \mu({}^{v_0}P_w)+\mu(\partial F)\\
&\equiv& \mu({}^{v_0}P_w)+d\mu(F)\\
&\equiv& \mu({}^{v_0}P_w)+\omega_2(F).
\end{eqnarray*}

But $\omega_2(F)$ is the number of 2-faces in $F$, which is $r$, the weight of the codeword $w$, which is even.  Therefore $\mu({}^{v_1}P_w) \equiv \mu({}^{v_0}P_w)\pmod{2}$.

Case 2: $w_m=1$

In this case, $j_n=m$ for some $n$, and we have the following diagram.
\begin{center}
\begin{picture}(320,80)(-10,-10)
\put(0,0){\line(1,0){60}}
\put(0,40){\line(1,0){60}}
\put(80,0){\line(1,0){120}}
\put(80,40){\line(1,0){120}}
\put(220,0){\line(1,0){60}}
\put(220,40){\line(1,0){60}}
\put(0,0){\line(0,1){40}}
\put(40,0){\line(0,1){40}}
\put(100,0){\line(0,1){40}}
\put(140,0){\line(0,1){40}}
\put(180,0){\line(0,1){40}}
\put(240,0){\line(0,1){40}}
\put(280,0){\line(0,1){40}}
\put(70,0){\makebox[0in][c]{$\cdots$}}
\put(210,20){\makebox[0in][c]{$\cdots$}}
\put(0,0){\circle*{4}}
\put(0,40){\circle*{4}}
\put(140,0){\circle*{4}}
\put(140,40){\circle*{4}}
\put(280,0){\circle*{4}}
\put(280,40){\circle*{4}}
\put(-4,-10){\makebox[0in][r]{$v_0$}}
\put(284,-10){\makebox[0in][l]{$v_1$}}
\put(-4,44){\makebox[0in][r]{$v_1$}}
\put(284,44){\makebox[0in][l]{$v_0$}}
\put(20,-10){\makebox[0in][c]{$j_1$}}
\put(260,-10){\makebox[0in][c]{$j_r$}}
\put(20,44){\makebox[0in][c]{$j_1$}}
\put(120,-10){\makebox[0in][c]{$j_{n-1}$}}
\put(120,44){\makebox[0in][c]{$j_{n-1}$}}
\put(160,-10){\makebox[0in][c]{$j_{n+1}$}}
\put(160,44){\makebox[0in][c]{$j_{n+1}$}}
\put(260,44){\makebox[0in][c]{$j_r$}}
\put(-4,16){\makebox[0in][r]{$m$}}
\put(36,16){\makebox[0in][r]{$m$}}
\put(96,16){\makebox[0in][r]{$m$}}
\put(136,16){\makebox[0in][c]{$m=j_n$}}
\put(184,16){\makebox[0in][l]{$m$}}
\put(244,16){\makebox[0in][l]{$m$}}
\put(284,16){\makebox[0in][l]{$m$}}
\put(140,-12){\makebox[0in][c]{$Q$}}
\put(140,44){\makebox[0in][c]{$P$}}
\end{picture}
\end{center}
Again, use Lemma~\ref{lemma:squarelifting} iteratively to map this figure into $A$ with the bottom left vertex going to $v_0$, and the edges going to edges colored with the labels shown.  The leftmost edge connects $v_0$ to $v_1$.

The path ${}^{v_0}P_w$ is in the image of the path that goes from the lower left corner along the lower edge until the point $Q$, then going up to $P$, then going across the top edge to the point in the upper right.  Similarly, the path ${}^{v_1}P_w$ is in the image of the path that goes along the top from the upper left to $P$, then down to $Q$, then along the bottom edge to the lower right.  Again, let $F$ be the sum of the faces that this diagram gets sent to in $A$.  As before,
\[
\mu({}^{v_1}P_w)\equiv\mu({}^{v_0}P_w)+\omega_2(F)\pmod{2},
\]
but $\omega_2(F)$, the number of 2-faces in $F$, is now $r-1$, which is odd.  Therefore $\mu({}^{v_1}P_w)\equiv\mu({}^{v_0}P_w)+1$.

General case: We return to the situation where we consider words $x$ with $\wt(x)>1$.  We iterate the above procedure for every $j$ so that $x_j=1$.  Note that $\rho_x$ is the composition of the $\rho_j$ for which $x_j=1$.  The result follows.
\end{proof}

\subsection{Node choice symmetry}
We have seen how the $k$-tuple ${}^{v_0}s(\mu)=(s_1,\ldots,s_k)$ captures the properties of $\mu$ that are invariant under vertex switches.  But it depends on the choice of a vertex $v_0$, as we have just seen.  These vertices can sometimes be easily distinguished: for instance, some are bosons and others are fermions; some have many arrows pointing away from it, others have few; and so on.  But in other cases, there is no property that can distinguish them, and these are node choice symmetries.

\begin{definition}
A node choice symmetry of an Adinkra $A$ is a permutation of the vertices and edges in $A$ that preserves
\begin{enumerate}
\item the incidence relation between edges and vertices,
\item the bipartition $(V_0,V_1)$ of the vertices,
\item the colors of the edges, and
\item the orientation of the arrows on the edges.
\end{enumerate}
That is, it preserves all the features of the Adinkra except for the dashedness.
\end{definition}

Let us assume $A$ is connected.  By the color-regular property of Adinkras, criteria 1 and 3 in the definition guarantees that the permutation is determined uniquely by where it sends a given vertex $v_0$ (to, say, $v_1$).  The path from $v_0$ to $v_1$ can be described as a word $x\in (\zt)^N$.  The permutation given by $\rho_x$ satisfies criterion 1 and 3, and by uniqueness, must be the original permutation.  Therefore, every node choice symmetry is given by $\rho_x$ for some $x$.

Criterion 2 is equivalent to $\wt(x)$ being even, and criterion 4 can be described in terms of height assignments by saying that $\rho_x(v)$ and $v$ have the same height (equivalently, the same engineering dimension) for every vertex $v$.  This motivates the definition of the following code:

\begin{definition}
The node choice code $\ch$ of a connected Adinkra is the set of words $x\in(\zt)^N$ so that for all vertices $v$ in the Adinkra, $\rho_x(v)$ has the same engineering dimension as $v$.
\end{definition}

Note that $\ch$ is a group under addition modulo 2, and is therefore a linear code.  Suppose $A$ is a quotient of $I^N$ by the doubly even code $C$.  Then $C$ is a subgroup of $\ch$.  The map $x\mapsto \rho_x$ describes an action of $\ch$ on the Adinkra via node choice symmetries, and $C$ is the subgroup that acts trivially on the Adinkra.  All node choice symmetries are obtained in this way.\footnote{In \cite{rAT2} we defined the node choice group to be $\ch$, and a node choice symmetry to be an element of $\ch$, whereas by our definition, a node choice symmetry is given by the actual permutation $\rho_x$, which is determined by an element in $\ch/C$.  This distinction will not be relevant for our purposes.}  Criterion 2 implies $\ch$ is an even code.  There is a generating matrix $H$ whose rows are the generating words for $\ch$.

The point here is that by considering the invariant $k$-tuple ${}^{v_0}s(\mu)$, one might be under the impression that there are $2^k$ odd dashings on $A=I^N/C$, where $k$ is the dimension of the code $C$.  But the node choice symmetries act on the vertices, and therefore on the possible $k$-tuples by changing $v_0$.  If we identify the various possible $k$-tuples, the result will give us invariants that identify the odd dashings on an Adinkra modulo vertex switches.

\begin{theorem}
Let $A$ be a quotient of $I^N$ by a code $C$ with generating matrix $G$. 
Let $H$ be the generating matrix for the node choice code.  Let $\mu$ be an odd dashing on $A$, and let $v_0$, $v_1$ be vertices of $A$.  There exists a node choice symmetry from $v_0$ to $v_1$ if and only if there exists a column vector $y$ so that
\[
GH^Ty = {}^{v_1}s(\mu)-{}^{v_0}s(\mu).
\]
\end{theorem}

\begin{proof}
If $\rho_x$ is a node choice symmetry, then since the columns of $H^T$ generate the node choice group, $x$ can be written as $H^Ty$ for some column vector $y$.  The result then follows from Theorem~\ref{thm:switchv0}.
\end{proof}

\begin{example}
Consider the $N=6$ example given below.  This is an Adinkra but the arrows have been removed for the sake of clarity.  Instead, the vertices have been placed at heights so that every arrow is assumed to point upward along its edge.  The fact that this can be done is explained in \cite{r6-1}.

\begin{center}
\begin{picture}(180,100)(0,0)
\put(0,0){\includegraphics[height=100pt]{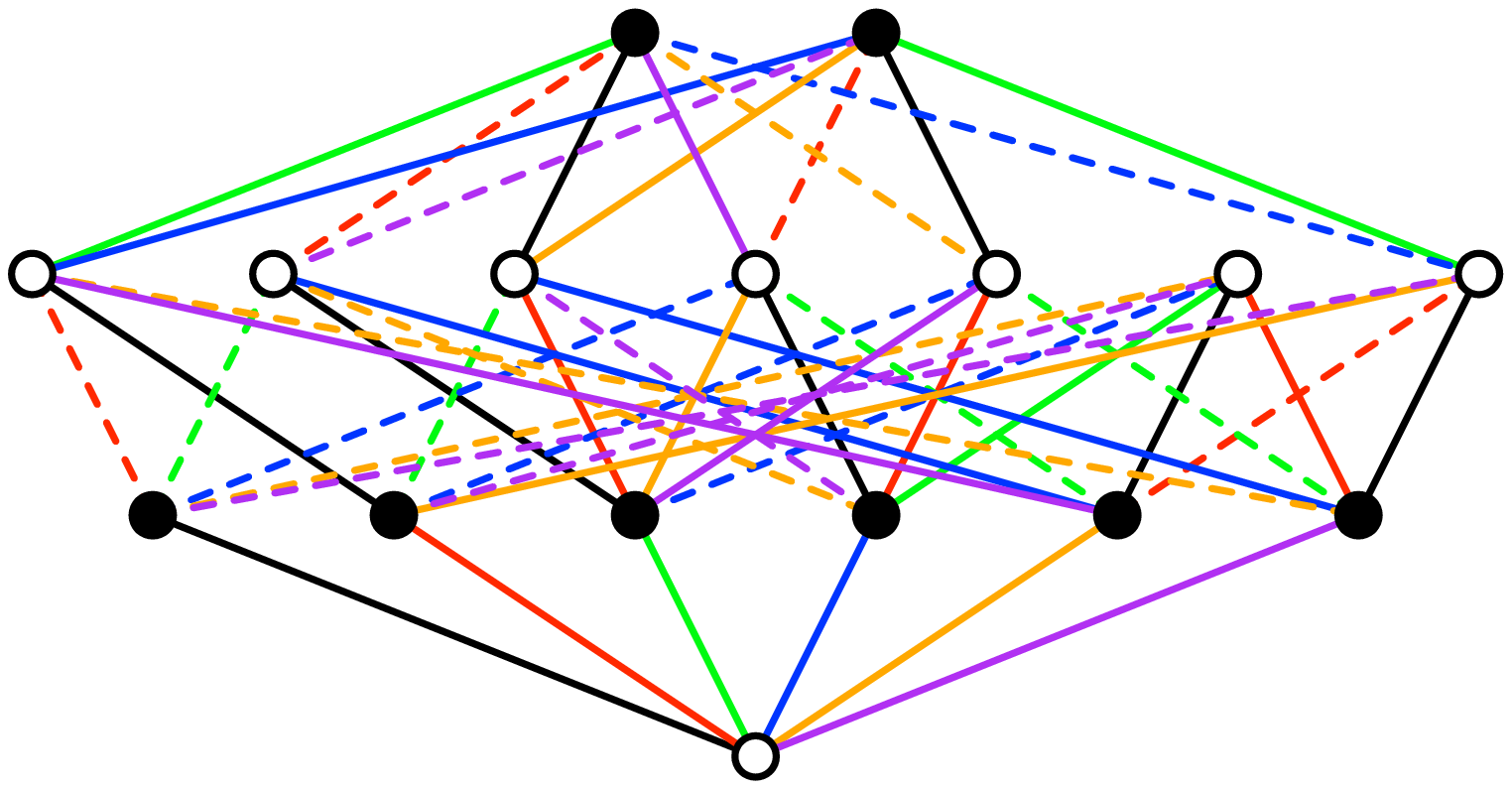}}
\put(88,-1){\makebox[0in][c]{$A$}}
\end{picture}
\end{center}
\begin{center}
Color key:

\begin{tabular}{c|c}
$j$&$j$th color\\\hline
$1$&black\\
$2$&orange\\
$3$&red\\
$4$&purple\\
$5$&blue\\
$6$&green
\end{tabular}
\end{center}

The code is $d_6$, generated by $111100$ and $001111$.  The generating matrix is thus

\[
G=\left[\begin{array}{cccccc}
1&1&1&1&0&0\\
0&0&1&1&1&1
\end{array}\right].
\]

Since $d_6$ is two-dimensional, there should be $2^2=4$ different dashings.  The matrix $G$ is not in reduced row-echelon form, but since we have $\left[\begin{array}{c}1\\ 0\end{array}\right]$ and 
$\left[\begin{array}{c}0\\ 1\end{array}\right]$
as colums, we can see that the other 3 dashings can be obtained by edge sign flips with color 1, color 5, and with both, as illustrated below.

\begin{center}
\begin{picture}(60,72)(0,-10)
\put(0,0){\includegraphics[height=50pt]{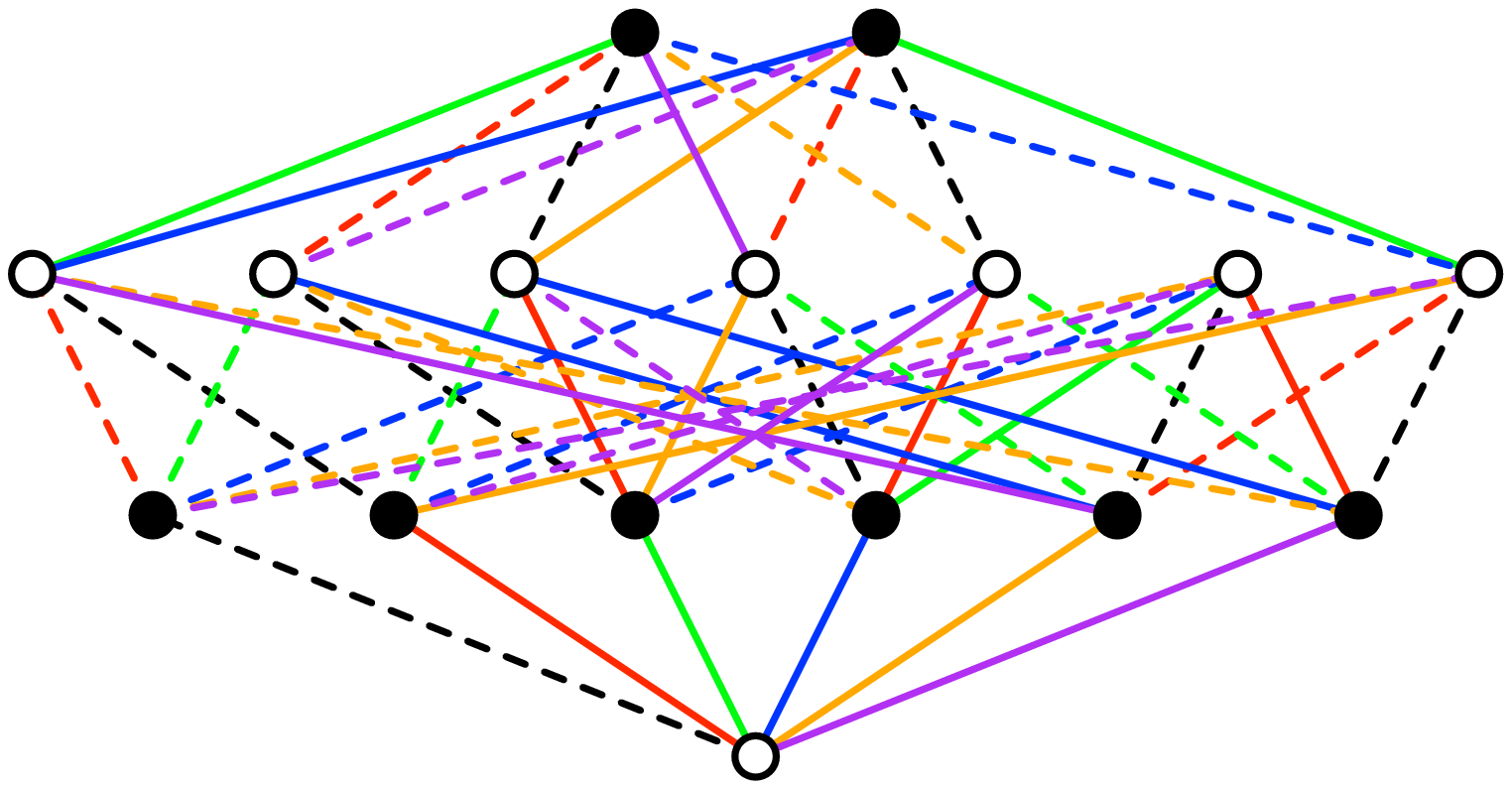}}
\end{picture}
\makebox[.5in]{}
\begin{picture}(60,72)(0,-10)
\put(0,0){\includegraphics[height=50pt]{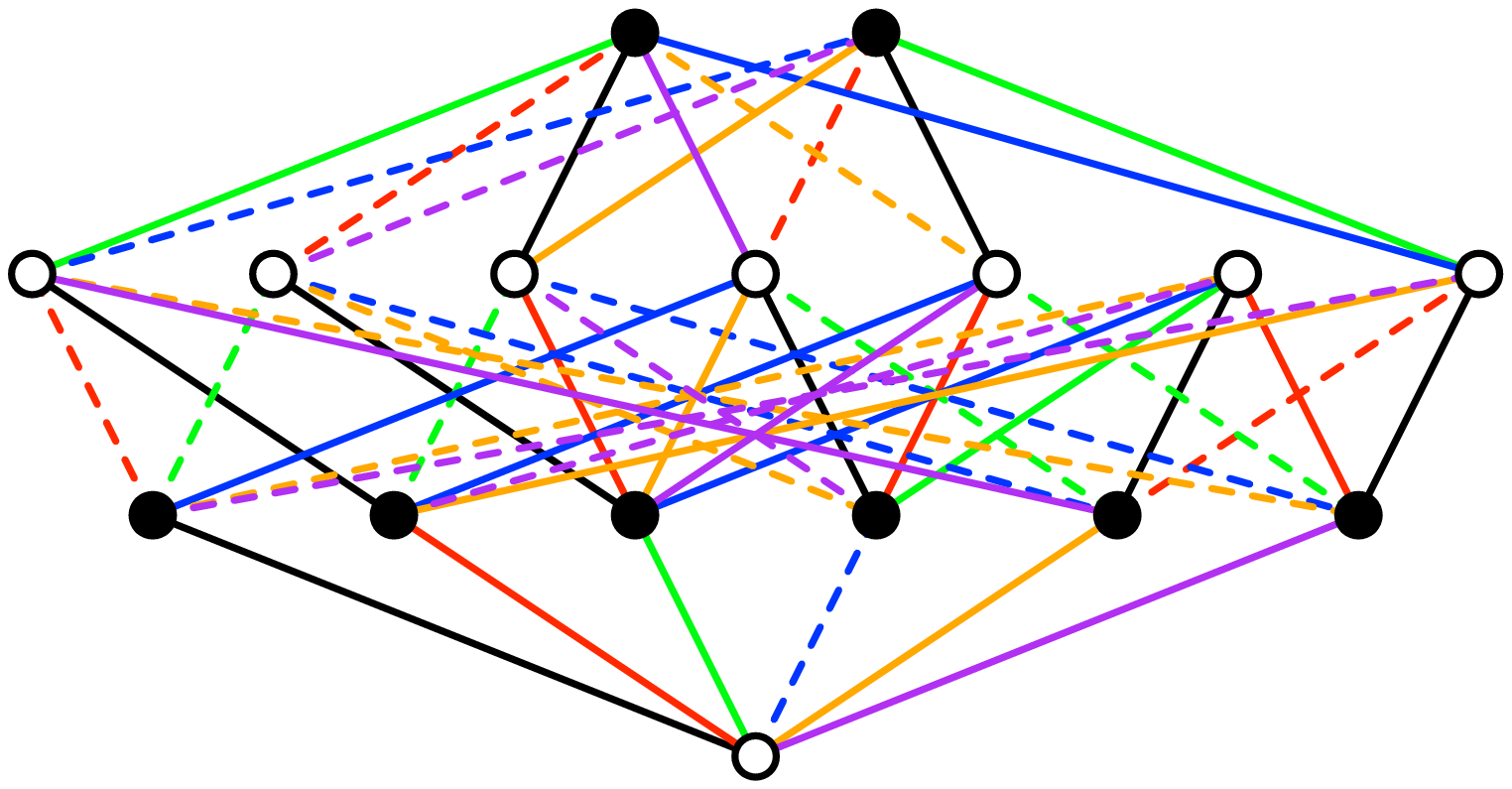}}
\end{picture}
\makebox[.5in]{}
\begin{picture}(60,72)(0,-10)
\put(0,0){\includegraphics[height=50pt]{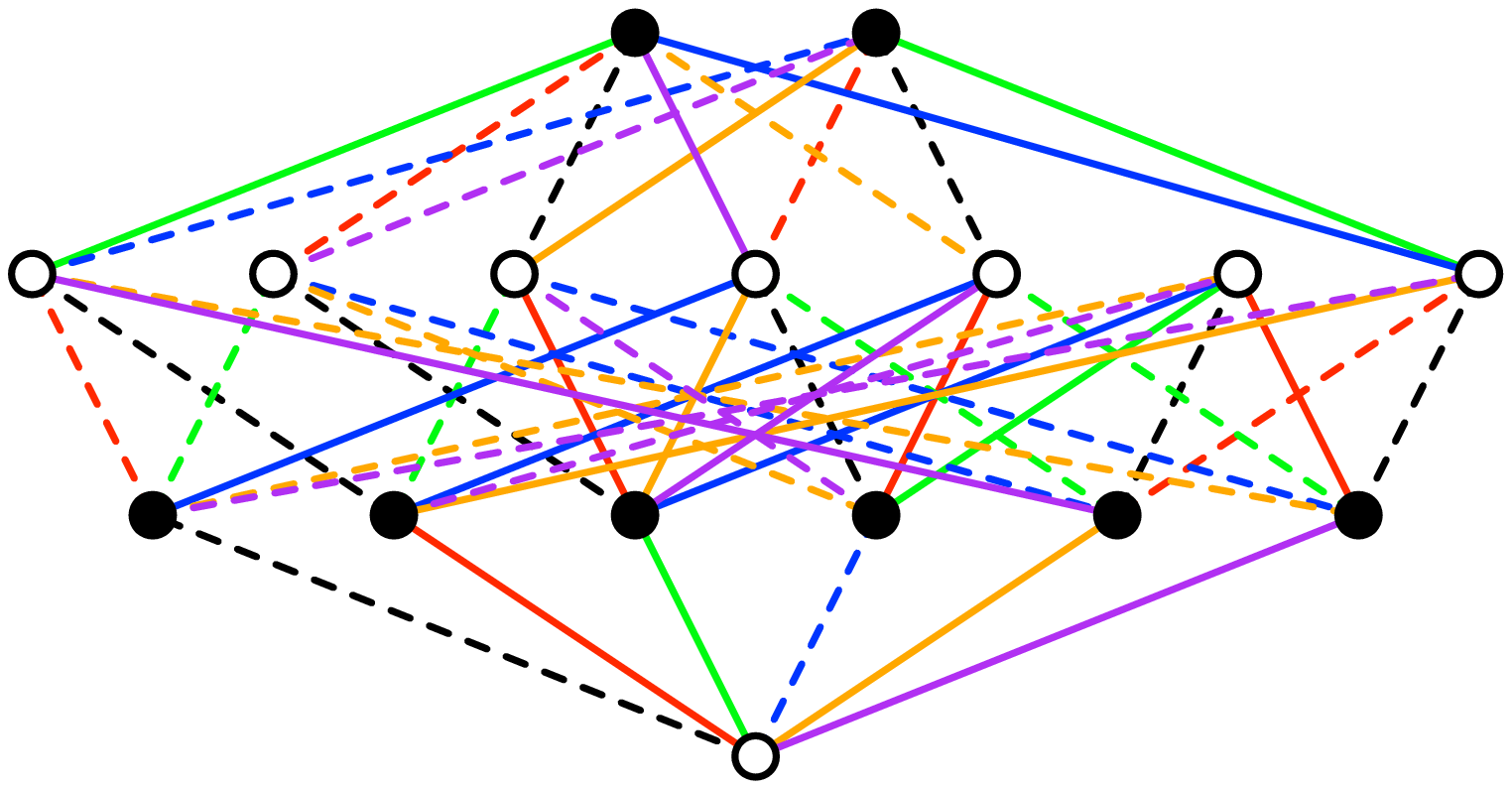}}
\end{picture}

\end{center}

There are no node choice symmetries, since any such symmetry must send the bottom node to itself.  Because a node choice symmetry is determined by where it sends any one vertex, the only node choice symmetry is the identity.

Take $v_0$ to be the lowest vertex, labelled $A$ in the diagram.  We trace a path with colors 1, 2, 3, then 4 (in that order) from $v_0$ and count the number of dashed edges, modulo 2.  This is $s_1$.  We do the same with the colors 3, 4, 5, then 6 to get $s_2$.
\end{example}

\begin{example}\label{ex:d62}
Now consider a different hanging of the same $N=6$ Adinkra:
\begin{center}
\begin{picture}(290,100)(0,0)
\put(0,0){\includegraphics[height=100pt]{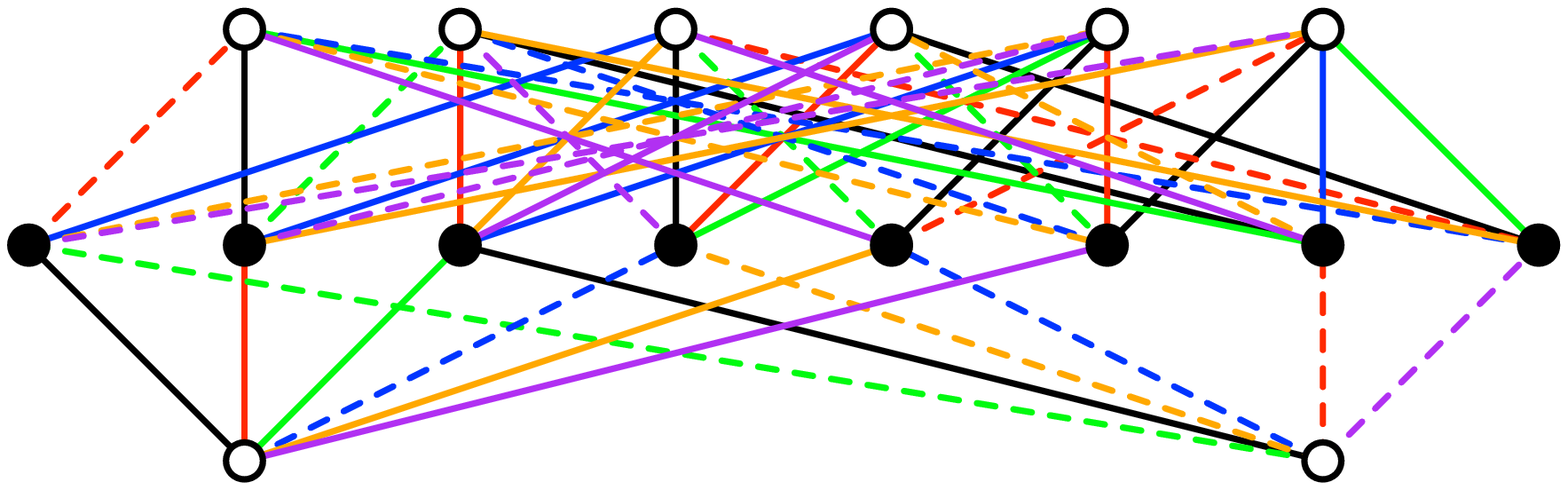}}
\put(46,0){$A$}
\put(240,0){$B$}
\end{picture}
\end{center}
In this case, $G$ is the same, but there is now a node choice symmetry that sends $A$ to $B$.  This can be described as a path from $A$ to $B$ following color 1 (black), then color 6 (green).  Hence, the matrix $H$ is
\[
H=\left[\begin{array}{cccccc}
1&1&1&1&0&0\\
0&0&1&1&1&1\\
1&0&0&0&0&1
\end{array}\right].
\]
The first two rows are simply the rows in $G$, to acknowledge that $C$ is a subgroup of $\ch$; they will be irrelevant anyway.  We compute:
\[
GH^T=\left[\begin{array}{ccc}
0&0&1\\
0&0&1
\end{array}\right].
\]
Therefore the effect of a node choice symmetry on $(s_1,s_2)$ is to either do nothing, or to toggle both $s_1$ and $s_2$.  We can thus use $s_1+s_2$ as a complete invariant of dashings, modulo vertex switches and node choice symmetries.
\end{example}

\begin{example}
We now consider the same $N=6$ Adinkra, but hung as a valise; that is, all bosons on one level, and all fermions on the other.  In this case, the arrows go from bosons to fermions, which in \cite{rA} was termed a {\em base Adinkra}.\footnote{A similar example could be obtained if the arrows go from fermions to bosons; this is the Klein flip of a base Adinkra.  A base Adinkra or its Klein flip is called a Valise Adinkra.}
\begin{center}
\begin{picture}(330,95)(-30,-10)
\put(0,0){\includegraphics[height=70pt]{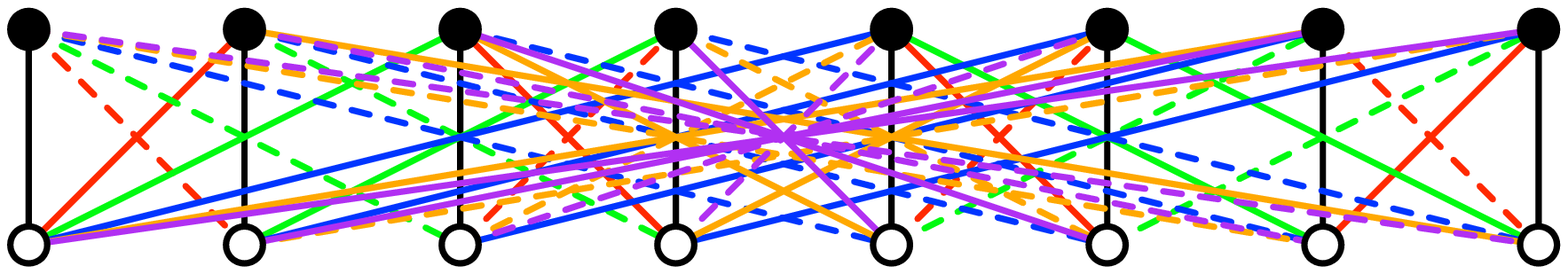}}
\put(-8,-10){\makebox[0in][r]{$(s_1,s_2)$}:}
\put(-8,65){\makebox[0in][r]{$(s_1,s_2)$}:}
\put(14,0){\makebox[0in][c]{$\phi_1$}}
\put(14,-10){\makebox[0in][c]{(1,0)}}
\put(55,0){\makebox[0in][c]{$\phi_2$}}
\put(55,-10){\makebox[0in][c]{(1,1)}}
\put(96,0){\makebox[0in][c]{$\phi_3$}}
\put(96,-10){\makebox[0in][c]{(0,1)}}
\put(137,0){\makebox[0in][c]{$\phi_4$}}
\put(137,-10){\makebox[0in][c]{(0,0)}}
\put(178,0){\makebox[0in][c]{$\phi_5$}}
\put(178,-10){\makebox[0in][c]{(0,1)}}
\put(219,0){\makebox[0in][c]{$\phi_6$}}
\put(219,-10){\makebox[0in][c]{(0,0)}}
\put(260,0){\makebox[0in][c]{$\phi_7$}}
\put(260,-10){\makebox[0in][c]{(1,0)}}
\put(301,0){\makebox[0in][c]{$\phi_8$}}
\put(301,-10){\makebox[0in][c]{(1,1)}}
\put(14,75){\makebox[0in][c]{$\psi_1$}}
\put(14,65){\makebox[0in][c]{(0,0)}}
\put(55,75){\makebox[0in][c]{$\psi_2$}}
\put(55,65){\makebox[0in][c]{(0,1)}}
\put(96,75){\makebox[0in][c]{$\psi_3$}}
\put(96,65){\makebox[0in][c]{(1,1)}}
\put(137,75){\makebox[0in][c]{$\psi_4$}}
\put(137,65){\makebox[0in][c]{(1,0)}}
\put(178,75){\makebox[0in][c]{$\psi_5$}}
\put(178,65){\makebox[0in][c]{(1,1)}}
\put(219,75){\makebox[0in][c]{$\psi_6$}}
\put(219,65){\makebox[0in][c]{(1,0)}}
\put(260,75){\makebox[0in][c]{$\psi_7$}}
\put(260,65){\makebox[0in][c]{(0,0)}}
\put(301,75){\makebox[0in][c]{$\psi_8$}}
\put(301,65){\makebox[0in][c]{(0,1)}}
\end{picture}
\end{center}
Given any two bosons, there is a node choice symmetry that takes one to the other.  Thus, the node choice group is the set of even weight words.  This is generated by words of weight 2 where the 1s are adjacent.
\[
H=\left[\begin{array}{cccccc}
1&1&0&0&0&0\\
0&1&1&0&0&0\\
0&0&1&1&0&0\\
0&0&0&1&1&0\\
0&0&0&0&1&1
\end{array}\right].
\]
Then we have
\[
GH^T=\left[\begin{array}{ccccc}
0&0&0&1&0\\
0&1&0&0&0
\end{array}\right],
\]
and note that the columns span the whole space $(\zt)^2$.  In particular, a node symmetry given by $y=(0,0,0,1,0)$, corresponding to $x=H^Ty=(0,0,0,1,1,0)$, involving $\rho_4$ followed by $\rho_5$, will change the value of $s_1$ without changing $s_2$.  Likewise, $y=(0,1,0,0,0)$, corresponding to $x=H^Ty=(0,1,1,0,0,0)$, involving $\rho_2$ followed by $\rho_3$, will change the value of $s_2$ without changing $s_1$.  And doing both changes both $s_1$ and $s_2$.  Therefore, the valise Adinkra for $d_6$ has only one dashing up to node choice symmetries.
\end{example}

\begin{example}\label{ex:d4s}
In the case $N=4$ with the code $d_4$ generated by $1111$, consider a valise Adinkra.  There are two dashings, shown below.
\begin{center}
\begin{picture}(120,90)(0,-10)
\put(0,0){\includegraphics[height=60pt]{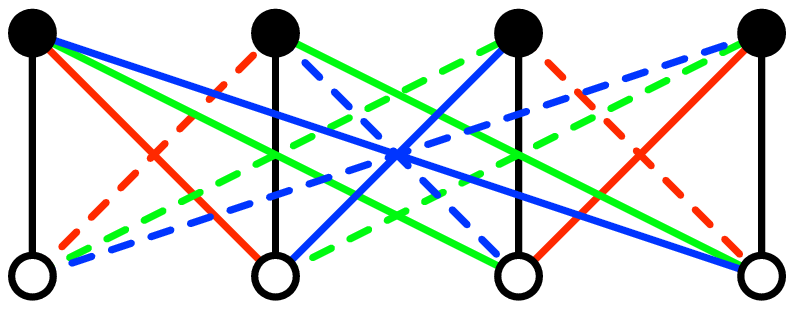}}
\put(12,0){\makebox[0in][c]{$\phi_1$}}
\put(47,0){\makebox[0in][c]{$\phi_2$}}
\put(82,0){\makebox[0in][c]{$\phi_3$}}
\put(117,0){\makebox[0in][c]{$\phi_4$}}
\put(12,-10){\makebox[0in][c]{(0)}}
\put(47,-10){\makebox[0in][c]{(0)}}
\put(82,-10){\makebox[0in][c]{(0)}}
\put(117,-10){\makebox[0in][c]{(0)}}
\put(12,67){\makebox[0in][c]{$\psi_1$}}
\put(12,57){\makebox[0in][c]{(1)}}
\put(47,67){\makebox[0in][c]{$\psi_2$}}
\put(47,57){\makebox[0in][c]{(1)}}
\put(82,67){\makebox[0in][c]{$\psi_3$}}
\put(82,57){\makebox[0in][c]{(1)}}
\put(117,67){\makebox[0in][c]{$\psi_4$}}
\put(117,57){\makebox[0in][c]{(1)}}
\put(-14,30){$A_1$:}
\end{picture}
\makebox[1in]{}
\begin{picture}(120,90)(0,-10)
\put(0,0){\includegraphics[height=60pt]{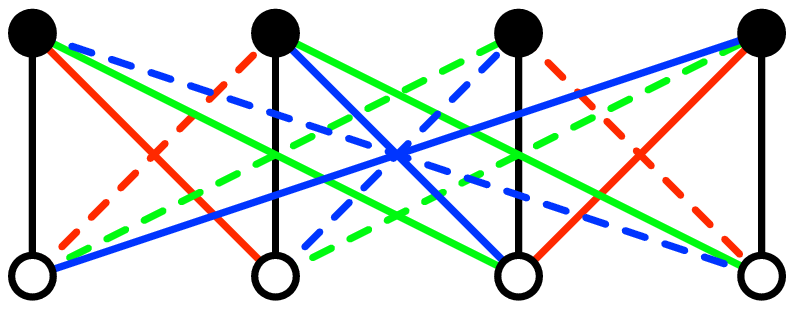}}
\put(12,0){\makebox[0in][c]{$\phi_1$}}
\put(47,0){\makebox[0in][c]{$\phi_2$}}
\put(82,0){\makebox[0in][c]{$\phi_3$}}
\put(117,0){\makebox[0in][c]{$\phi_4$}}
\put(12,-10){\makebox[0in][c]{(1)}}
\put(47,-10){\makebox[0in][c]{(1)}}
\put(82,-10){\makebox[0in][c]{(1)}}
\put(117,-10){\makebox[0in][c]{(1)}}
\put(12,67){\makebox[0in][c]{$\psi_1$}}
\put(12,57){\makebox[0in][c]{(0)}}
\put(47,67){\makebox[0in][c]{$\psi_2$}}
\put(47,57){\makebox[0in][c]{(0)}}
\put(82,67){\makebox[0in][c]{$\psi_3$}}
\put(82,57){\makebox[0in][c]{(0)}}
\put(117,67){\makebox[0in][c]{$\psi_4$}}
\put(117,57){\makebox[0in][c]{(0)}}
\put(-14,30){$A_2$:}
\end{picture}
\end{center}
\vspace{10pt}
\begin{center}
Color key:

\begin{tabular}{c|c}
$j$&$j$th color\\\hline
$1$&black\\
$2$&red\\
$3$&green\\
$4$&blue
\end{tabular}
\end{center}
The numbers in the parentheses are the values of $s_1(\mu)$ at each vertex.  For the Adinkra $A_1$ on the left, every boson has $s_1=0$ and every fermion has $s_1=1$.  We do an edge sign flip on any edge (in this case, blue) and obtain the Adinkra $A_2$ on the right, where every boson has $s_1=0$ and every fermion has $s_1=0$.  Since every node choice symmetry preserves the bipartition of the vertices, no node choice symmetry can swap the two, and so $A_1$ and $A_2$ are really different.  Note that $GH^T$ is
$\left[\begin{array}{ccc}
0&0&0
\end{array}\right]$,
which confirms this result.  If we do a Klein flip on $A_1$ (swap $V_0$ and $V_1$ in the bipartition), then do vertex lifts so that the bosons are still on the bottom of the diagram, the result will be (up to vertex switches) $A_2$.

The distinction between the two dashings is the same as that between the chiral and the twisted chiral superfield, except that the chiral and twisted chiral superfields are not Valise Adinkras.
\end{example}

\begin{example}
Consider the valise Adinkra for $N=8$ with $C=e_8$:
\begin{center}
\begin{picture}(240,104)(-25,-20)
\put(0,0){\includegraphics[height=60pt]{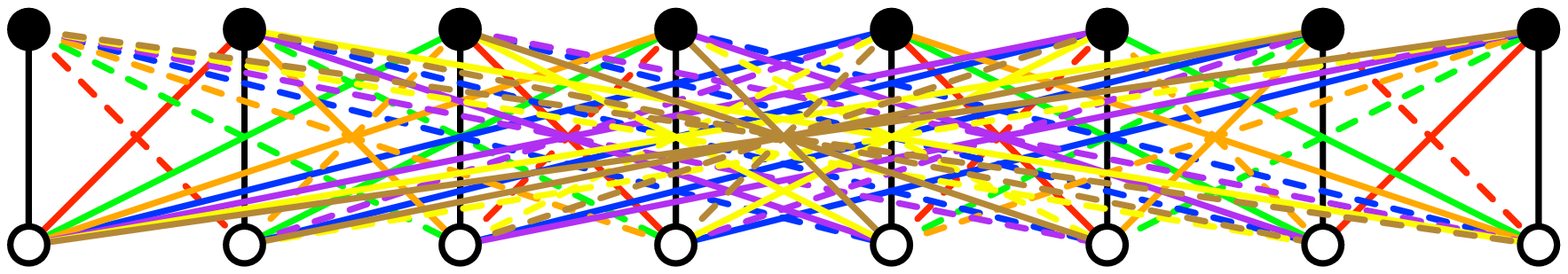}}
\put(-16,-20){\makebox[0in][r]{$s_1+s_3$}:}
\put(-16,54){\makebox[0in][r]{$s_1+s_3$}:}
\put(-16,-10){\makebox[0in][r]{$(s_1,s_2,s_3,s_4)$}:}
\put(-16,64){\makebox[0in][r]{$(s_1,s_2,s_3,s_4)$}:}
\put(12,0){\makebox[0in][c]{$\phi_1$}}
\put(12,-10){\makebox[0in][c]{(1,1,0,1)}}
\put(12,-20){\makebox[0in][c]{1}}
\put(47,0){\makebox[0in][c]{$\phi_2$}}
\put(47,-10){\makebox[0in][c]{(1,1,0,0)}}
\put(47,-20){\makebox[0in][c]{1}}
\put(82,0){\makebox[0in][c]{$\phi_3$}}
\put(82,-10){\makebox[0in][c]{(1,0,0,1)}}
\put(82,-20){\makebox[0in][c]{1}}
\put(117,0){\makebox[0in][c]{$\phi_4$}}
\put(117,-10){\makebox[0in][c]{(1,0,0,0)}}
\put(117,-20){\makebox[0in][c]{1}}
\put(152,0){\makebox[0in][c]{$\phi_5$}}
\put(152,-10){\makebox[0in][c]{(0,0,1,1)}}
\put(152,-20){\makebox[0in][c]{1}}
\put(187,0){\makebox[0in][c]{$\phi_6$}}
\put(187,-10){\makebox[0in][c]{(0,0,1,0)}}
\put(187,-20){\makebox[0in][c]{1}}
\put(222,0){\makebox[0in][c]{$\phi_7$}}
\put(222,-10){\makebox[0in][c]{(0,1,1,1)}}
\put(222,-20){\makebox[0in][c]{1}}
\put(257,0){\makebox[0in][c]{$\phi_8$}}
\put(257,-10){\makebox[0in][c]{(0,1,1,0)}}
\put(257,-20){\makebox[0in][c]{1}}
\put(12,74){\makebox[0in][c]{$\psi_1$}}
\put(12,64){\makebox[0in][c]{(0,1,0,0)}}
\put(12,54){\makebox[0in][c]{0}}
\put(47,74){\makebox[0in][c]{$\psi_2$}}
\put(47,64){\makebox[0in][c]{(0,1,0,1)}}
\put(47,54){\makebox[0in][c]{0}}
\put(82,74){\makebox[0in][c]{$\psi_3$}}
\put(82,64){\makebox[0in][c]{(0,0,0,0)}}
\put(82,54){\makebox[0in][c]{0}}
\put(117,74){\makebox[0in][c]{$\psi_4$}}
\put(117,64){\makebox[0in][c]{(0,0,0,1)}}
\put(117,54){\makebox[0in][c]{0}}
\put(152,74){\makebox[0in][c]{$\phi_5$}}
\put(152,64){\makebox[0in][c]{(1,0,1,0)}}
\put(152,54){\makebox[0in][c]{0}}
\put(187,74){\makebox[0in][c]{$\phi_6$}}
\put(187,64){\makebox[0in][c]{(1,0,1,1)}}
\put(187,54){\makebox[0in][c]{0}}
\put(222,74){\makebox[0in][c]{$\phi_7$}}
\put(222,64){\makebox[0in][c]{(1,1,1,0)}}
\put(222,54){\makebox[0in][c]{0}}
\put(257,74){\makebox[0in][c]{$\phi_8$}}
\put(257,64){\makebox[0in][c]{(1,1,1,1)}}
\put(257,54){\makebox[0in][c]{0}}
\end{picture}
\end{center}
\vspace{10pt}
\begin{center}
Color key:

\begin{tabular}{c|c}
$j$&$j$th color\\\hline
$1$&black\\
$2$&red\\
$3$&green\\
$4$&orange\\
$5$&blue\\
$6$&purple\\
$7$&yellow\\
$8$&brown
\end{tabular}
\end{center}

Here the generating matrix for $e_8$ is
\[
G=\left[\begin{array}{cccccccc}
1&1&1&1&0&0&0&0\\
0&0&1&1&1&1&0&0\\
0&0&0&0&1&1&1&1\\
1&0&1&0&1&0&1&0
\end{array}\right].
\]
The node choice code has generating matrix
\[
H=\left[\begin{array}{cccccccc}
1&1&0&0&0&0&0&0\\
0&1&1&0&0&0&0&0\\
0&0&1&1&0&0&0&0\\
0&0&0&1&1&0&0&0\\
0&0&0&0&1&1&0&0\\
0&0&0&0&0&1&1&0\\
0&0&0&0&0&0&1&1
\end{array}\right],
\]
and we compute
\[
GH^T=\left[\begin{array}{ccccccc}
0&0&0&1&0&0&0\\
0&1&0&0&0&1&0\\
0&0&0&1&0&0&0\\
1&1&1&1&1&1&1
\end{array}\right].
\]
Note that the first and third rows are the same, but other than that, the rows are linearly independent.  Thus, node choice symmetries can produce $2^3=8$ different $s(\mu)$ sequences.  These correspond to the 8 bosons: note in the figure that the $s(\mu)$ sequences are listed at each vertex.  As long as the symmetry sends bosons to bosons, we cannot alter $s_1+s_3$, which is truly invariant.  But other non-trivial linear combinations of the $s_a$ can always be changed.  Hence, there are two distinct dashings on an $N=8$ valise Adinkra that are a quotient of $I^8$ by the code $e_8$.

The other dashing can be obtained by edge flipping any color, for instance $Q_8$ (brown):
\begin{center}
\begin{picture}(240,104)(-25,-20)
\put(0,0){\includegraphics[height=60pt]{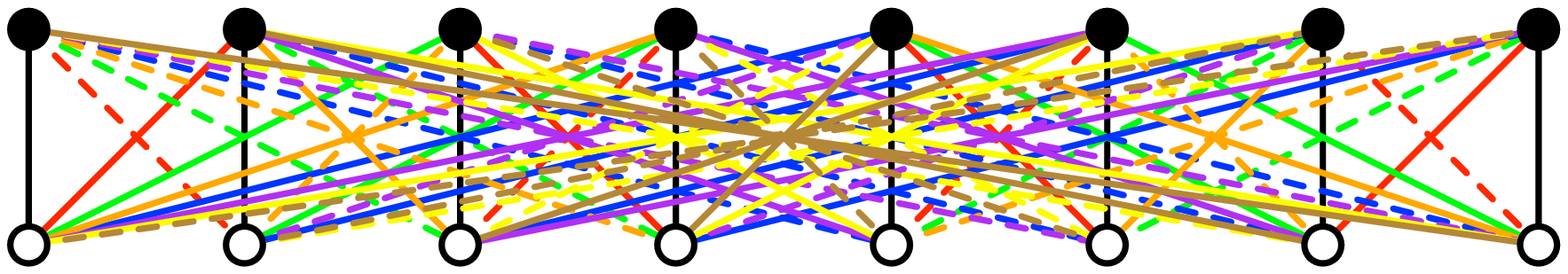}}
\put(-16,-20){\makebox[0in][r]{$s_1+s_3$}:}
\put(-16,54){\makebox[0in][r]{$s_1+s_3$}:}
\put(-16,-10){\makebox[0in][r]{$(s_1,s_2,s_3,s_4)$}:}
\put(-16,64){\makebox[0in][r]{$(s_1,s_2,s_3,s_4)$}:}
\put(12,0){\makebox[0in][c]{$\phi_1$}}
\put(12,-10){\makebox[0in][c]{(1,1,1,1)}}
\put(12,-20){\makebox[0in][c]{0}}
\put(47,0){\makebox[0in][c]{$\phi_2$}}
\put(47,-10){\makebox[0in][c]{(1,1,1,0)}}
\put(47,-20){\makebox[0in][c]{0}}
\put(82,0){\makebox[0in][c]{$\phi_3$}}
\put(82,-10){\makebox[0in][c]{(1,0,1,1)}}
\put(82,-20){\makebox[0in][c]{0}}
\put(117,0){\makebox[0in][c]{$\phi_4$}}
\put(117,-10){\makebox[0in][c]{(1,0,1,0)}}
\put(117,-20){\makebox[0in][c]{0}}
\put(152,0){\makebox[0in][c]{$\phi_5$}}
\put(152,-10){\makebox[0in][c]{(0,0,0,1)}}
\put(152,-20){\makebox[0in][c]{0}}
\put(187,0){\makebox[0in][c]{$\phi_6$}}
\put(187,-10){\makebox[0in][c]{(0,0,0,0)}}
\put(187,-20){\makebox[0in][c]{0}}
\put(222,0){\makebox[0in][c]{$\phi_7$}}
\put(222,-10){\makebox[0in][c]{(0,1,0,1)}}
\put(222,-20){\makebox[0in][c]{0}}
\put(257,0){\makebox[0in][c]{$\phi_8$}}
\put(257,-10){\makebox[0in][c]{(0,1,0,0)}}
\put(257,-20){\makebox[0in][c]{0}}
\put(12,74){\makebox[0in][c]{$\psi_1$}}
\put(12,64){\makebox[0in][c]{(0,1,1,0)}}
\put(12,54){\makebox[0in][c]{1}}
\put(47,74){\makebox[0in][c]{$\psi_2$}}
\put(47,64){\makebox[0in][c]{(0,1,1,1)}}
\put(47,54){\makebox[0in][c]{1}}
\put(82,74){\makebox[0in][c]{$\psi_3$}}
\put(82,64){\makebox[0in][c]{(0,0,1,0)}}
\put(82,54){\makebox[0in][c]{1}}
\put(117,74){\makebox[0in][c]{$\psi_4$}}
\put(117,64){\makebox[0in][c]{(0,0,1,1)}}
\put(117,54){\makebox[0in][c]{1}}
\put(152,74){\makebox[0in][c]{$\phi_5$}}
\put(152,64){\makebox[0in][c]{(1,0,0,0)}}
\put(152,54){\makebox[0in][c]{1}}
\put(187,74){\makebox[0in][c]{$\phi_6$}}
\put(187,64){\makebox[0in][c]{(1,0,0,1)}}
\put(187,54){\makebox[0in][c]{1}}
\put(222,74){\makebox[0in][c]{$\phi_7$}}
\put(222,64){\makebox[0in][c]{(1,1,0,0)}}
\put(222,54){\makebox[0in][c]{1}}
\put(257,74){\makebox[0in][c]{$\phi_8$}}
\put(257,64){\makebox[0in][c]{(1,1,0,1)}}
\put(257,54){\makebox[0in][c]{1}}
\end{picture}
\end{center}
\end{example}

\begin{example}
Suppose more generally we have any connected Valise Adinkra $A=I^N/C$.  The node choice code is the set $E$ of even weight words in $(\zt)^N$.  The matrix $H$ is then
\[
H=\left[\begin{array}{ccccccc}
1&1&0&0&\cdots&0&0\\
0&1&1&0&\cdots&0&0\\
0&0&1&1&\cdots&0&0\\
\vdots&&&&&&\vdots\\
0&0&0&0&\cdots&1&1
\end{array}\right].
\]
Let the $i$th row of $H$ be denoted $h_i$.

Consider the generating matrix $G$ for the code $C$.  Let $z$ be the length $N$ word $(1,\ldots,1)$.

The dimension $D$ of the span of the columns of $GH^T$ can be determined by doing Gauss--Jordan elimination.  This dimension is $k$ minus the number of rows of $0$s of the reduced row-echelon form of $GH^T$.  Thus, $D=k$ if and only if no non-trivial linear combination of the rows of $GH^T$ is zero.  Such a linear combination can be done on $G$, resulting in a non-trivial linear combination of the generating words, i.e., a codeword $w\in C$, with $wH^T=0$ (here $w$ is written as a row matrix).  But unless $w=0$ or $w=z$, there will be a row of $H$, $h_i$, so that
\[\sum_{j=1}^N w_j h_{i,j}\equiv 1\pmod{2}.\]
To see this, simply locate a column of $w$ where $w$ changes from $0$ to $1$ or vice versa.  So the only non-trivial codeword $w$ with $wH^T=0$ is $z$.

This demonstrates that if $z\not\in C$, then $D=k$, and therefore any $s(\mu)$ can be turned into $0$ by choosing $v_0$ appropriately.  In other words, if $z\not\in C$, then there is only one odd dashing up to node choice symmetries and vertex switches.

If $z\in C$, then since there is only one non-trivial $w$ with $\sum_{j=1}^N w_j h_{i,j} \equiv 1\pmod{2}$, there is only one row of $0$s in the reduced row-echelon form of $GH^T$.  Thus, $D=k-1$.  Choose the generating set for $C$ so that $z$ is one of the generating words, say, the first one.  Then $s_1(\mu)$ is not altered by any node choice symmetry (since all elements of $\ch$ are even, they do nothing to $z$), but the other $s_a$ for $a>1$ can be turned into $0$ by choosing $v_0$ appropriately.  In other words, if $z\in C$, then there are two odd dashings up to node choice symmetries and vertex switches.

Since $C$ is doubly even, a necessary condition for $z\in C$ is that $N$ is a multiple of 4.  And when $N$ is a multiple of 4, if $C$ is maximal, then $z\in C$, since otherwise, the span of $C$ and $z$ forms a larger doubly even code.

This is in accordance to the following observations.  First, valise representations of the $N$-extended super Poincar\'e algebra are in one-to-one correspondence with supermodules of the Clifford algebra $\Cl(0,N)$, which in turn are in one-to-one correspondence with representations of $\Cl(0,N+1)$\cite{rAT2}.  Second, when $N$ is a multiple of 4, there are exactly two isomorphism classes of such irreducible representations, and otherwise, there is just one such\cite{LM}.

\end{example}

\subsection{Dashing invariants}
The results of the previous section indicate that the sequence $(s_1,\ldots,s_k)$ does not in general work as an invariants for a dashing, because of node choice symmetries.  But the methods of the previous section also suggest a solution: find linear combinations of the $s_i$ that are invariant under every node choice symmetry.

\begin{algorithm}
Input: code generator matrix G and Node choice code generator matrix $H$

Output: A set of linear combinations of the $s_i$ that is invariant under node choice symmetries

Procedure:

Create an augmented matrix, consisting of $GH^T$ on the left, and a $k\times k$ identity matrix on the right.

Perform Gauss--Jordan elimination on this augmented matrix.

For each row of the resulting matrix where the left side has only zeros, we will create an invariant.  On such a row, look on the right side of the matrix.  For each column $i$ on the right with a 1, include a summand $s_i$.
\end{algorithm}

\begin{example}
Take the $N=6$ Example~\ref{ex:d62} above.  Here we have
\[
GH^T=\left[\begin{array}{ccc}
0&0&1\\
0&0&1
\end{array}\right],
\]
which we augment with an identity matrix, as demanded in the algorithm:
\[
\left[\begin{array}{ccc|cc}
0&0&1&1&0\\
0&0&1&0&1
\end{array}\right].
\]
Doing Gauss--Jordan elimination gives
\[
\left[\begin{array}{ccc|cc}
0&0&1&1&0\\
0&0&0&1&1
\end{array}\right],
\]
and there is one row which has $0$ on the left of the line: the second row.  The right side of this is
\[
\begin{array}{cc}
1&1,
\end{array}
\]
which indicates that $s_1+s_2$ is an invariant of both node symmetries and vertex switches.

\end{example}

\section{Other applications of cubical cohomology to Adinkras}
This paper has been about odd dashings.  But cubical cohomology has other applications to the mathematics of Adinkras.

\subsection{Bipartitions}
The existence and classification of bipartitions is analogous to that of odd dashings, but where the role of $w_2$ is replaced by $w_1$.  Just as a dashing gives rise to a 1-cochain, a partition of the vertices is equivalent to a 0-cochain $\mu$, where bosons are assigned 0, and fermions are assigned 1.  This is a bipartition if and only if every edge is incident with one boson and one fermion, i.e., if for every edge $e$, $\mu(\partial e)=1$.  This is equivalent to saying that $d\mu=\omega_1$.  Since $d\omega_1=0$, this is possible if and only if $w_1=0$ in $H^1(A;\zt)$.

This choice of $\mu$ is unique up to an element of $H^0(A;\zt)$.  If $A$ is connected, then there are only two choices, up to swapping all bosons and fermions (called a Klein flip in \cite{rA}).  If $A$ is not connected, then a Klein flip could be done independently in each connected component.

\subsection{Arrows}
If we define cohomology of an Adinkra not only with $\zt$-coefficients, but with $\bz$-coefficients, then an orientation (that is, drawing an arrow on every edge) on a cubical graph gives rise to a one-cochain $\mu$ in $C^1(A;\bz)$ in the following manner: for every oriented edge $e$ going from $v_p$ to $v_q$, we define $\mu(e)$ to be 1 if the arrow points from $v_p$ to $v_q$, and $-1$ otherwise.  Then the orientation is non-escheric if and only if $d\mu=0$.  Since $H^1(A;\bz)=0$, a non-escheric orientation $\mu$ can be described as $df$ for some 0-cochain $f$.  Such a cochain is an assignment of an integer to each vertex, which if we imagine is a placement of the vertices at different ``heights'' on the page.  Since $\mu$ takes its values in $\pm 1$, it must be that edges connect vertices of adjacent height.  This height assignment is unique up to adding an integer additive constant for each connected component of the Adinkra.

This height function provides what in particle physics is called engineering dimension or mass dimension.  Quantities in physics come with units of time, distance, mass, or combinations of these.  If we choose units so that the speed of light $c=1$ and Planck's constant $\hbar=1$, then units can be written as a power of the unit of mass.  This power is the engineering dimension.  Assuming the transformation rules do not involve constants that have dimension, the function that takes each field and returns twice the engineering dimension is a height assignment.

A non-cohomological proof of these facts, and an application of these ideas to classifying all non-escheric orientations, is given in \cite{r6-1}.

The escheric central charge example in \cite{rA} admits a similar 1-form with values in $\bz_4$.

\def\rasp{\leavevmode\raise.45ex\hbox{$\rhook$}}

\end{document}